\documentclass[draftcls, onecolumn, journal]{IEEEtran}
\usepackage{ifpdf}
\usepackage{cite} 
 
\ifCLASSINFOpdf 
  \usepackage[pdftex]{graphicx} 
  \DeclareGraphicsExtensions{.pdf} 
\else 
  \usepackage[dvips]{graphicx} 
  \DeclareGraphicsExtensions{.eps} 
\fi 
 
\usepackage{tikz} 
\usetikzlibrary{arrows,shapes,backgrounds} 
 
\usepackage[cmex10]{amsmath} 
\usepackage{amssymb} 
\usepackage{array}

\ifpdf 
\usepackage[pdftex]{hyperref} 
    \hypersetup{ 
    colorlinks=true,
    citecolor=black,
    filecolor=black,
    linkcolor=black,
    urlcolor=black 
  }   
\else 
\usepackage[active]{srcltx} 
\fi 
 
\newtheorem{theorem}{Theorem} 
\newtheorem{corollary}{Corollary} 
\newtheorem{lemma}{Lemma} 
\newtheorem{proposition}{Proposition} 
\newtheorem{definition}{Definition}

\DeclareMathOperator{\E}{E} 
\DeclareMathOperator{\Gr}{Gr} 
\DeclareMathOperator{\Pj}{Pj} 
\DeclareMathOperator{\rank}{\text{rk}} 
\DeclareMathOperator{\Fr}{Fr}

\newcommand{\floor}[1]{\lfloor #1 \rfloor} 
\newcommand{\etal}{\textit{et al.}} 
\newcommand{\subs}{\text{SS}} 
\newcommand{\csub}{\text{C-SS}} 
 
\newcommand{\ffield}{\mathbb{F}}

\newcommand{\lspan}[1]{\langle #1 \rangle} 
\newcommand{\tr}{\top} 
 
\newcommand{\gcos}[2]{\begin{pmatrix} #1 \\ 
#2 \end{pmatrix}_q} 
\newcommand{\gco}[2]{\left({}^{#1}_{#2}\right)_q} 
\newcommand{\cmat}[2]{\chi^{#1}_{#2}} 
\newcommand{\cmatt}[2]{\zeta^{#1}_{#2}} 
\newcommand{\bigO}{\mathcal{O}}

\newcommand{\bX}{\mathbf{X}} 
\newcommand{\bY}{\mathbf{Y}} 
 
\newcommand{\bH}{\mathbf{H}} 
\newcommand{\bzero}{\mathbf{0}}

\newcommand{\loc}{\text{LOC}} 
 
\newcommand{\lc}{\mathcal{C}} 
\newcommand{\ow}{o.w.}

\newcounter{locequator}[section] 
\renewcommand{\thelocequator}{\alph{locequator}} 
\newcommand{\locequa}[1]{\stepcounter{locequator}\stackrel{\text{(\alph{locequator})}}{#1}} 
 
\begin{document} 
\title{On Linear Operator Channels over Finite Fields}

\author{\IEEEauthorblockN{Shenghao Yang\IEEEauthorrefmark{1}, Siu-Wai
    Ho\IEEEauthorrefmark{2}, Jin Meng\IEEEauthorrefmark{3},
    En-hui Yang\IEEEauthorrefmark{3} and Raymond W. Yeung\IEEEauthorrefmark{1}} \\
  \IEEEauthorblockA{\IEEEauthorrefmark{1} Institute of Network Coding
    and Department of Information
    Engineering\\ The Chinese University of Hong Kong, Hong Kong SAR \\ Email: \{shyang,whyeung\}@ie.cuhk.edu.hk}  \\
  \IEEEauthorblockA{\IEEEauthorrefmark{2} Institute for
    Telecommunications Research \\
    University of South Australia, Australia \\Email:  siuwai.ho@unisa.edu.au }\\
  \IEEEauthorblockA{ \IEEEauthorrefmark{3}
    Department of Electrical and Computer Engineering \\
    University of Waterloo, Canada \\Emails: \{j4meng,
    ehyang\}@uwaterloo.ca} }

\maketitle

\begin{abstract} 
  Motivated by linear network coding, communication channels perform 
  linear operation over finite fields, namely linear operator channels 
  (LOCs), are studied in this paper. For such a channel, its output 
  vector is a linear transform of its input vector, and the 
  transformation matrix is randomly and independently generated. The 
  transformation matrix is assumed to remain constant for every $T$ 
  input vectors and to be unknown to both the transmitter and the 
  receiver. There are NO constraints on the distribution of 
  the transformation matrix and the field size. 
 
  Specifically, the optimality of subspace coding over LOCs is 
  investigated. A lower bound on the maximum achievable rate 
  of subspace coding is obtained and  
  it is shown to be tight for some cases. The maximum achievable rate 
  of constant-dimensional subspace coding is characterized and  
  the loss of rate incurred by using constant-dimensional 
  subspace coding is insignificant. 
 
  The maximum achievable rate of channel training is close to the 
  lower bound on the maximum achievable rate of subspace coding.  Two 
  coding approaches based on channel training are proposed and their 
  performances are evaluated. Our first approach makes use of 
  rank-metric codes and its optimality depends on the existence of 
  maximum rank distance codes. Our second approach applies linear 
  coding and it can achieve the maximum achievable rate of channel 
  training. Our code designs require only the knowledge of the 
  expectation of the rank of the transformation matrix. The second 
  scheme can also be realized ratelessly without a priori knowledge of 
  the channel statistics. 
\end{abstract} 
 

\section{Introduction}

Let $\ffield$ be a finite field with $q$ elements.  A \emph{linear 
  operator channel (LOC)} with input $X\in \ffield^{T\times M}$ and 
output $Y\in \ffield^{T\times N}$ is given by 
\begin{equation}\label{eq:model} 
  Y = XH, 
\end{equation} 
where $H$ is called the transformation matrix.

Our motivation to study LOCs comes from linear network 
coding, a research topic that has drawn extensive interest in the 
past ten years. Linear network coding is a network transmission 
technique that can achieve the capacity of multicasting in 
communication networks \cite{flow,linear,alg,poly}.  
Different from routing, 
linear network coding allows network nodes to relay new packets 
generated by linear combinations.  The point-to-point transmission of 
a network employing linear network coding is given by a LOC, where $H$ is the model of network transfer matrix and depends on the network topology \cite{linear,alg}.

A recent 
research topic where LOCs have found applications is the deterministic 
model of wireless networks \cite{avest07a,avest07b}.  This 
deterministic model provides a good approximation of certain wireless 
network behaviors and has shown its impact on the study of wireless 
networks. When employing linear operations in intermediate network 
nodes, the point-to-point transmission of the deterministic model of 
wireless networks is also given by a LOC \cite{ebrahimi09, 
  ebrahimi10}.

Even though some aspects of LOCs have been well studied in linear network coding, our understanding of LOCs is far from enough.  
In fact, the only case that LOCs are completely understood is that $H$ has a constant rank $M$.  
However, $H$ in general can have rank deficiency (i.e., $\rank(H)<M$) due to the change of network topology, link failure, packet loss, and so on. Even without these network related dynamics, $H$ has a random rank when random linear network coding is applied where new packets are generated by random linear combinations.  
Towards more sophisticated applications of linear network coding, a systematic study of LOCs  becomes necessary.  In this work, we study the information theoretic communication limits of 
LOCs with a general distribution of $H$ and discuss coding 
for LOCs.

\subsection{Some Related Works} 
 
We review some works of linear network coding that related to our discussions.

When both the transmitter and the receiver know the  
instances of $H$, the 
transmission through a LOC is called the \emph{coherent transmission}. 
For a network with fixed and known topology, linear network codes can be designed 
deterministically in polynomial time \cite{poly}. The transmission through such a network is usually 
assumed to be coherent.  For the coherent transmission, the rank of $H$ 
determines the capability of information transmission and it is bounded by the maximum flow form the transmitter  to the receiver \cite{linear, alg, avest07a,avest07b}. 
 
In communication networks where  the network topology is dynamic and/or 
unknown, e.g., wireless communication networks, deterministic design of network coding is difficult to realize. Random linear network coding is an efficient approach to apply 
network coding in such communication networks \cite{ho06j, gkant05,dimakis07, 
  fragouli08, xiao09}.  The transformation matrix of a communication network employing random linear network coding, called a \emph{random linear coding network 
  (RLCN)}, is a random matrix and its instances are assumed to be unknown in both the transmitter and the receiver. Such a kind of transmission is referred to as the 
\emph{noncoherent transmission}.  The existing works on the noncoherent transmission of RLCN considers several special distributions of $H$.

In various models and applications of random linear network coding 
\cite{ho06j,montanari07, silva08c, chachu07,Katti08}, $H$ is assumed 
to be an invertible square matrix\footnote{More generally, the assumption is that $H$ has rank $M$, which implies $N\geq M$.}. This assumption is based on the 
fact that when $H$ is a square matrix, i.e., $M=N$, it is full rank with high 
probability if i) $M$ is less than or equal to the 
maximum flow from the transmitter to the receiver, and ii) the field 
size for network coding is sufficiently large comparing with the 
number of network nodes \cite{ho06j, balli07}. However, random linear network coding with small 
finite fields is attractive for low computing complexity. For example, wireless sensor networks is 
characterized by large network size and limited computing capability 
of network nodes. Using large finite field operations in sensors may not be a good choice. Moreover, the maximum flow varies due to the dynamic of wireless networks. For these reasons, full rank transformation matrices cannot be assumed in many applications.

 K\"otter and Kschischang \cite{koetter08j} introduced a model of 
 random linear network coding, called K\"otter-Kschischang operator 
 channel (or KK operator channel), that takes vector spaces as input 
 and output, and commits fixed dimension erasures and additive 
 errors. Their model considers a special kind of rank-deficiency of $H$ 
 that gives fixed dimension erasures, defined as the difference of the 
 dimension of the output and input vector spaces. 
They introduced \emph{subspace 
  coding} for random linear network coding that can be used to correct 
erasures, defined as the rank difference between the output and input 
matrices, as well as additive errors \cite{koetter08j}.  Silva \etal\ \cite{silva08j} 
constructed (unit-block) subspace codes using rank-metric codes 
\cite{gabidulin85}, called \emph{unit-block lifted rank-metric codes} 
here, which are nearly optimal in terms of achieving a Singleton type 
bound of (unit-block) subspace codes \cite{koetter08j}.  The coding 
scheme proposed by Ho \etal \cite{ho06j} for random linear network 
coding is a special case of unit-length lifted rank-metric codes for 
the transmission without erasures and errors.

Jafari \etal \cite{siav08, jafari09} studied $H$ containing uniformly i.i.d. components---such a 
matrix is called a \emph{purely random matrix}.  However, there is no rigorous 
justification of why purely random matrices can 
reflect the properties of general random linear network 
coding. Moreover, the problem-specific techniques used to analyze 
purely matrices are difficult to be extended to the general 
cases.

\subsection{Summary of Our Work} 
 
In this paper, we study LOCs without any constraints on the 
distribution of $H$.  The purely 
random transformation matrix and the invertible transformation matrix 
are special cases in our problem.  We allow the transformation matrix has 
arbitrary rank and contains correlated 
components. 
We do not assume large finite fields to guarantee that the rank of $H$ 
is full rank with high probability.  
We mainly consider the noncoherent transmission of LOCs by assuming the instances of $H$ is unknown in both the transmitter and the receiver. 
 
Our results can be applied to (random) linear 
network coding in both wireless and wireline networks without 
constraints on the network topology and the field size, as long as the 
input and output of the network can be modelled by a LOC.  For 
example, link failures and packets losses, which do not change the 
linear relation between the input and output, can be taken into 
consideration.  But the network transformation can also suffers from 
random errors and malicious modifications, for which we have to model 
the network transformation as 
\begin{equation*} 
  Y = X H + Z, 
\end{equation*} 
and there is no equivalent way to model it as a LOC.  
We do not consider nonzero $Z$ as discussed in \cite{koetter08j, montanari07, silva08c}.

Our results are summarized as follows. 
 
We generalize the concept of subspace coding in \cite{koetter08j} to multiple usages of a LOC and study its achievable rates. 
Let $\bar C$ be the capacity of a LOC and let $\bar C_{\subs}$ 
be the maximum achievable rate of subspace coding for a 
LOC. We obtain that $(1-M/T)\E[\rank(H)] + \epsilon(T,q) \leq \bar 
C_{\subs} \leq \bar C \leq\E[\rank(H)]$, where $\E[\rank(H)]$ is the 
expectation of the rank of $H$ and $0<\epsilon(T,q) < 1.8/(T\log_2 
q)$. 
Moreover, we show that $\bar C_{\subs}=\bar C$ for 
\emph{uniform} LOCs, a class of LOCs that includes the purely random 
transformation matrix and the invertible transformation matrix studied 
in \cite{silva08c, siav08, jafari09}.

An unknown transformation matrix is \emph{regular} if 
its rank can take any value from zero to $M$.  A LOC is \emph{regular} 
if its transformation matrix is regular.  For {regular} LOCs with 
sufficiently large $T$, we prove that the lower bound on $\bar C_{\subs}$ is tight, and 
$\bar C_{\subs}$ is achieved by the $M$-dimensional subspace coding.  For 
example, a purely random $H$ with $M\leq N$ is uniform and 
regular. Thus $M$-dimensional subspace coding achieves its capacity 
when $T$ is sufficiently large.

Moreover, $\bar C_{\subs}$ can be well approximated by subspace codes using 
subspaces with the same dimension, called \emph{constant-dimensional subspace codes}.   
Let $\bar C_{\csub}$ be the 
maximum achievable rate of constant-dimensional subspace 
coding.  We show that $\bar C_{\subs} - \bar C_{\csub}< (\log_2 
\min\{M,N\})/(T\log_2q)$, which is much smaller than $\bar 
C_{\subs}$ for practical channel parameters. 
For general LOCs, we find the optimal dimension $r^*$ 
such that there exists an $r^*$-dimensional subspace code achieving 
$\bar C_{\csub}$.  Taking the LOCs with an invertible 
$H$ as an example, $M$ is the optimal dimension when $T\geq 2M+1$.

Channel training is a coding scheme for LOCs that uses parts of its 
input matrix to recover the instance of $H$.  The maximum achievable 
rate of using channel training $\bar C_{\text{CT}}$ is 
$(1-M/T)\E[\rank(H)]$, which is very close to the lower bound of $\bar 
C_{\subs}$.  We further proposed extended channel training codes to 
reducing the overhead of channel training codes. We give upper and 
lower bounds on the maximum achievable rate of extended channel training 
codes and show the gap between bounds is small. 
 
The coding scheme proposed by Ho \etal \cite{ho06j} and the unit-block 
lifted rank-metric codes proposed by Silva \etal \cite{silva08j} fall 
in the class of channel training.  We show that unit-block lifted 
rank-metric codes can achieve $\bar C_{\text{CT}}$ only when $H$ has a 
constant rank.  If $H$ have an arbitrary rank, the maximum 
achievable rate of unit-block lifted rank-metric codes is demonstrated 
to be far from $\bar C_{\text{CT}}$ for certain rank distribution of 
$H$.

To achieve $\bar C_{\text{CT}}$, we consider two coding schemes.  In 
the first scheme, we extend the method of Silva \etal\ \cite{silva08j} 
to construct codes for LOCs by multiple uses of the channel. The 
constructed code is called \emph{lifted rank-metric code}.  The 
optimality of lifted rank-metric codes, in the sense of achieving 
$\bar C_{\text{CT}}$, depends on the existence of the 
maximum-rank-distance (MRD) codes in classical algebraic coding 
theory, which was first studied in \cite{gabidulin85}.  Specifically, 
we show that if $T\gg M$, lifted rank-metric codes can approximately 
approach $\bar C_{\text{CT}}$. Otherwise, since the existence of MRD 
codes is unclear, it is uncertain if lifted rank-metric codes can 
achieve $\bar C_{\text{CT}}$.  Exsiting decoding algorithms of 
rank-metric codes can be applied to lifted rank-metric codes. The 
decoding complexity is given by $\bigO(n^2)$ field operations in 
$\ffield$, where $n$ is the block length of the codes.

We further propose a class of codes called \emph{lifted linear matrix 
  codes}, which can achieve $\bar C_{\text{CT}}$ for all $T\geq M$. 
We show that with probability more than half, a randomly choosen 
generator matrix gives good performance. We obain the error exponent 
of decoding lifted linear matrix codes.  The decoding of a lifted 
linear matrix code has complexity given by $\bigO(n^3)$ field 
operations when applying Gaussian elimination.  Lifted linear matrix 
codes can be realized ratelessly if the channel has a neglectable rate of 
feedback. 
 
Both lifted rank-metric codes and lifted linear matrix codes are 
universal in the sense that i) only the knowledge of $\E[\rank(H)]$ is 
required to design codes and ii) a code has similar performance for 
all LOCs with the same $\E[\rank(H)]$. Furthermore, rateless lifted linear 
matrix codes do not require any priori knowledge of channel statistics.

\subsection{Organization} 
 
This paper also provides a general framework to study LOCs.  Some 
notations and mathematical results that are used in our discussion, 
including some counting problems related to projective spaces, are 
introduced in \S\ref{sec:pre}.  Self-contained proofs of these 
counting problems are given in Appendix~\ref{sec:count}. In 
\S\ref{sec:mul}, linear operator channels are formally defined, and 
coherent and noncoherent transmission of LOCs are discussed.  In 
\S\ref{sec:ct} we give the maximum achievable rate of a noncoherent 
transmission scheme: channel training and study the bounds on the 
maximum achievable rate of extended channel training.  In 
\S\ref{sec:symm}, we reveal an intrinsic symmetric property of LOCs 
that holds for any distribution of the transformation matrix. These symmetric 
properties can help to determine the capacity-achieving input 
distributions of LOCs. In \S\ref{sec:subspace} and 
\S\ref{sec:input} we study subspace coding.  From \S\ref{sec:codes} to 
\S\ref{sec:linear}, two coding approaches for LOCs are introduced. 
The last section contains the concluding remarks.

\section{Preliminaries} 
\label{sec:pre} 
 
Let $\ffield$ 
be the finite field with $q$ elements, $\ffield^t$ be the 
$t$-dimensional vector space over $\ffield$, and 
$\ffield^{t\times m}$ be the set of all $t\times m$ matrices 
over $\ffield$. For a matrix $\bX$, let $\rank(\bX)$ be its 
rank, let $\bX^\tr$ be its transpose, and let $\lspan{\bX}$ 
be its column space, the subspace spanned by the column 
vectors of $\bX$. Similarly, the row space of $\bX$ is 
denoted by $\lspan{\bX^\tr}$. If $V$ is a subspace of $U$, 
we write $V\leq U$.

The \emph{projective space} $\Pj(\ffield^t)$ is the collection of 
all subspaces of $\ffield^t$. 
Let $\Pj(m,\ffield^t)$ be the subset of $\Pj(\ffield^t)$ 
that contains all the subspaces with dimension less than or 
equal to $m$. 
This paper involves some counting problems in projective 
space, which are discussed in Appendix~\ref{sec:count}. 
Let $\Fr(\ffield^{m\times r})$ be the set of full rank matrices 
in $\ffield^{m\times r}$. Define 
\begin{equation}\label{eq:111} 
\cmat{m}{r}=\left\{\begin{array}{ll} 
    (q^m-1)(q^m-q)\cdots(q^m-q^{r-1}) & 
    r>0 \\ 1 & r=0 \end{array} \right. 
\end{equation} 
for $r\leq m$. 
By {Lemma~\ref{lm:NumFullRankM}}, $|\Fr(\ffield^{m\times r})| 
= \cmat{m}{r}$. 
Define 
\begin{equation}\label{eq:speaker} 
  \cmatt{m}{r} = \cmat{m}{r} q^{-mr}. 
\end{equation} 
Since the number of $m\times r$ matrices is $q^{mr}$, 
$\cmatt{m}{r}$ can be regarded as the probability that a 
randomly chosen $m\times r$ matrix is full rank 
(ref. Lemma~\ref{lm:rk}). 
The \emph{Grassmannian} 
$\Gr(r,\ffield^t)$ is the set of all $r$-dimensional 
subspaces of $\ffield^t$. Thus $\Pj(m,\ffield^t) = \bigcup_{r\leq 
  m}\Gr(r,\ffield^t)$.  
The \emph{Gaussian binomials} are defined as 
\begin{equation*}  
  \gcos{m}{r} = \frac{\cmat{m}{r}}{\cmat{r}{r}}. 
\end{equation*}  
By {Lemma~\ref{lemma:g9ds}, $\gco{t}{r} = |\Gr(r,\ffield^t)|$.  
Let  
\begin{equation} \label{eq:23048fa} 
 \cmat{m,n}{r}= \frac{\cmat{m}{r}\cmat{n}{r}}{\cmat{r}{r}}, 
\end{equation} 
which is the number of $m\times n$ matrices with rank $r$ 
(see {Lemma~\ref{lm:NumMGivenRank}}).

For a discrete random variable (RV) $X$, we use $p_X$ to denote its 
probability mass function (PMF).  
Let $X$ and $Y$ be RVs over discrete alphabets $\mathcal{X}$ 
and $\mathcal{Y}$, respectively. We write a transition 
probability (matrix) from $\mathcal{X}$ to $\mathcal{Y}$ as $P_{Y|X}(\bX|\bY)$, 
$\bX\in\mathcal{X}$ and $\bY\in \mathcal{Y}$.  
{When the context is clear, we may omit the subscript of 
  $p_X$ and $P_{Y|X}$ to simplify the notations.}

\section{Linear Operator Channels} 
\label{sec:mul}

\subsection{Formulations} 
 
We first introduce a vector formulation of LOCs which reveals more 
details than the one given in \eqref{eq:model}.  Let $T$, $M$ and $N$ 
be nonnegative 
integers.  
A linear operator channel takes an $M$-dimensional vector as input and 
an $N$-dimensional vector as output.  The $i$th input $ x_i \in 
\ffield^{1\times M}$ and the $i$th output $ y_i\in \ffield^{1\times 
  N}$ are related by 
\begin{equation*} 
   y_i = x_i H_i, 
\end{equation*} 
where $H_i$ is a random matrix over $\ffield^{M\times N}$. 
We consider that $H_i$ keeps constant for $T$ consecutive input vectors, 
i.e., 
\begin{equation*} 
  H_{nT+1}=H_{nT+2}=\cdots=H_{nT+T},\quad n=0,1,2,\cdots; 
\end{equation*} 
and $H_{nT+1}$, $n=0,1,\cdots$, are independent and follow 
the same generic distribution of random variable $H$.   
By considering $T$ consecutive inputs/outputs as a matrix, we have the  
matrix formulation  
given in \eqref{eq:model}. 
Here, $T$ is called the \emph{inaction period}; $M\times N$ 
is called the \emph{dimension} of the 
LOC. 
A LOC with transformation matrix $H$ and inaction period $T$ is 
denoted by $\loc(H,T)$. Unless otherwise specified, we use the capital letters 
$M$ and $N$ for the dimension of $\loc(H,T)$. We will 
use the matrix formulation of the LOCs in this paper 
exclusively. When we talk about one use of 
$\loc(H,T)$, we mean the channel transmits one $T\times M$ matrix.

A communication network employing linear network coding can be 
modeled by a LOC. For example, when applying 
linear network coding in relay nodes, the transformation matrix of the 
network in Fig.~\ref{fig:ex} is 
\begin{equation}\label{eq:forme} 
  H = \begin{bmatrix}\alpha_1 & \alpha_2\beta_1 \\ 0 & \beta_2 \end{bmatrix}, 
\end{equation} 
in which $\alpha_1, \alpha_2, \beta_1$ and $\beta_2$ are linear 
combination coefficients taking value in $\ffield$. These coefficients 
can be fixed or random depending on the linear network coding 
approach. For a given network topology, the general formulation of the 
transformation matrix of linear network coding can be found in 
\cite{alg}.

For wireless networks without a centralized control, the transmission 
of network nodes is spontaneous and the network topology is 
dynamic. When employing random linear network coding, the inputs and 
the outputs of a wireless network still have linear relations 
\cite{chachu07}\footnote{We do not consider the encoding of packages 
  with errors}, but the formulation of the transformation matrix is 
difficult to obtain. The instances of the transformation matrix of 
random network coding is usually assumed to be unknown in both the 
transmitter and the receiver.  We will mainly discuss this kind of 
transmission of LOCs (see \S\ref{sec:nco}).

The transmission of random linear network coding is packetized. The 
source node organizes its data into $M$ packages, called a batch, and 
each of which contains $T$ symbols from $\ffield$. Network nodes 
perform linear network coding among the symbols in the same position 
of the packages in one batch, and the coding for all the positions are 
the same.  This packetized transmission matches our assumption that 
the transformation matrix keeps constant for $T$ consecutive input 
vectors. For this reason, the inaction period is also called the 
\emph{packet length}. 
The sink node try to collect $N$ (usually, $N\geq M$) packages in this 
batch to decode the original packages. This gives a physical meaning 
of the dimension of LOCs.

\begin{figure} 
	\centering  
	\tikzstyle{dot}=[circle,draw=gray!80,fill=gray!20,thick,inner sep=2pt,minimum size=10pt] 
	\begin{tikzpicture}[scale=1] 
	\node[dot] (s) at (0,0) {$s$}; 
	\node[dot] (t) at (0,-3) {$t$}; 
	\node[dot] (a) at (-0.9,-1) {$a$} edge [<-, bend 
        left=20] node[left,swap] {$x$} (s) edge [->, bend 
        right=25] node[left,swap] {$\alpha_1x$} (t); 
	\node[dot] (b) at (1,-2) {$b$} edge [<-, bend 
        right=35] node[right,swap] {$y$} (s) 
        edge [->, bend left=30] node[right,swap] 
        {$\alpha_2\beta_1x+\beta_2y$} (t) edge [<-, bend left=6] 
        node[right,swap] {$\alpha_2x$} (a); 
	\end{tikzpicture} 
        \caption{A directed network with the source node $s$ and the 
          sink node $t$. Each edge in the network is a communication 
          link that can transmit a symbol from $\ffield$ without error. 
          Node $a$ and $b$ are rely nodes that apply linear 
          network coding. The transmitted symbols through links are 
          labeled.} \label{fig:ex} 
\end{figure}
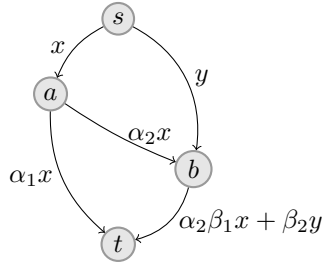

\subsection{Coherent Transmission of LOCs} 
\label{sec:co} 
 
We call the instances of the transformation matrix the \emph{channel 
  information (CI)}. The transmission with known CI at both the transmitter 
and the receiver is called \emph{coherent transmission}.  When the 
instance of $H$ is $\bH$, the maximum achievable rate of coherent 
transmission is $\max_{p_X} I(X;Y|H=\bH)$. Thus, the maximum 
achievable rate of coherent transmission (also called the 
\emph{coherent capacity}) is 
\begin{align*} 
  C_{\text{co}}(H,T) & = \sum_{\bH} p_H(\bH) \max_{p_X} 
  I(X;Y|H=\bH).
\end{align*} 
Unless otherwise specified, we use a base-$2$ logarithm in 
this paper so that $C_{\text{co}}(H,T)$ has a bit unit. 
 
Similar to coherent transmission, we can consider the transmission 
with CI only available at the receiver. We also assume that $X$ and $H$ are 
independent---this assumption is consistent with the transmitter does 
not know the instances of $H$. The maximum achievable 
rate of such transmission is 
\begin{equation*} 
  C_{\text{R-CI}}(H,T) = \max_{p_X} I(X;YH). 
\end{equation*} 
A random matrix is \emph{purely} random if it has uniformly 
independent components. 
 
\begin{proposition} \label{prop:1} $C_{\text{R-CI}}(H,T) = 
  C_{\text{co}}(H,T) = {T\log_2 q} \E[\rank(H)]$ and both capacities 
  are achieved by the purely random input distribution. 
\end{proposition} 
\begin{IEEEproof}
  We first consider the coherent transmission. We 
  know 
\begin{align*} 
  I(X;Y|H=\bH)  
  & = H(Y|H=\bH) - H(Y|X,H=\bH) \\ 
  & = H(Y|H=\bH). 
\end{align*} 
Let $x_i$ and $y_i$ be the $i$th rows of $X$ and $Y$, 
respectively. Since $y_i = x_i \bH$, i.e., $y_i$ is a vector in the subspace  
spanned by the row vectors of $\bH$, 
\begin{equation*} 
  H(y_i|H=\bH) \leq \log_2 q^{\rank(\bH)} = \rank(\bH) \log_2 q, 
\end{equation*} 
in which the equality is achieved when $x_i$ contains 
uniformly independent components.  
Hence,  
\begin{align*} 
  H(Y|H=\bH) & \leq  \sum_{i=1}^T H(y_i|H=\bH) \\ 
  & \leq \rank(\bH) T\log_2 q, 
\end{align*} 
where the first equality is achieved when $x_i$, $i=1,\cdots, T$, 
are independent. 
Therefore,  
\begin{align*} 
  C_{\text{co}}(H,T) & = \sum_{\bH} p_H(\bH) \max_{p_X} 
  I(X;Y|H=\bH) \\ & = \sum_{\bH}p_H(\bH) \rank(\bH) T\log_2 q \\ 
    & = \E[\rank(H)] T\log_2 q. 
\end{align*}

Now we consider the transmission 
with CI only available at the receiver.  We know 
\begin{align*} 
  I(X;YH) & = I(X;Y|H) + I(X;H) \\ 
  & = I(X;Y|H)\\ 
  & = H(Y|H) - H(Y|XH)\\ 
  & = H(Y|H), 
\end{align*} 
in which $I(X;H)=0$ since $X$ and $H$ are independent. 
Similar to the coherent case, 
\begin{align*} 
  H(Y|H) & = \sum_{\bH} p_H(\bH) H(Y|H=\bH) \\ 
  & \leq \sum_{i=1}^T \sum_{\bH} p_H(\bH) H(y_i|H=\bH) \\ 
  & \leq \sum_{i=1}^T \sum_{\bH} p_H(\bH) \rank(\bH) \log_2 q \\ 
  & = \E[\rank(H)] T\log_2 q, 
\end{align*} 
where the equality is achieved by $X$ with uniformly independent 
components.   
\end{IEEEproof}

\noindent\textbf{Remark:}  
Note that we do not assume $X$ and $H$ are independent for coherent 
transmission.  In fact for coherent transmission, the transmitter can 
use its knowledge of $\bH$ in encoding. Without lose of generality, we 
assume that the first $\rank(\bH)$ rows of $\bH$ are linearly 
independent.  So the transmitter can encode its information in an 
$M$-dimensional vector which contains only nonzero values in its first 
$\rank(\bH)$ components. The receiver can decode these nonzero values 
by solving a linear system of equations. Such scheme has transmission 
rate $\rank(\bH)T\log_2 q$, which achieves the coherent capacity. The 
coding that achieves $\E[\rank(H)] T\log_2 q$ with CI only 
available at the receiver, discussed in \S\ref{sec:codes}, is more 
involved.

\subsection{Noncoherent Transmission  of LOCs} 
\label{sec:nco} 
 
The transmission without the knowledge of CI in both the transmitter 
and the receiver is called \emph{noncoherent transmission}.  Same to 
the case with CI only available at the receiver, we assume that $H$ and $X$ 
are independent for noncoherent transmission. Under this assumption, 
\begin{align*} 
  p_{XY}(\bX,\bY) & = \Pr\{X=\bX,Y=\bY\}\\ 
  & = \Pr\{X=\bX,\bX H = \bY\} \\ 
  & = \Pr\{X=\bX\}\Pr\{\bX H =\bY\}. 
\end{align*} 
Thus, the transition probability $P_{Y|X}(\bY|\bX)$ of noncoherent 
transmission is given by 
\begin{equation} 
 P_{Y|X}(\bY|\bX)  = \Pr\{\bX H=\bY\}. \label{eq:trans} 
\end{equation} 
 
Unless otherwise specified, we consider noncoherent transmission of 
LOCs in the rest of this paper. For noncoherent transmission, a 
LOC is a \emph{discrete memoryless channel} (DMC).  The 
\emph{(noncoherent) capacity} of $\loc(H,T)$ is 
\begin{align*} 
 C(H,T)=\max_{p_X} I(X;Y). 
\end{align*}  
We also consider the normalized channel capacity 
\begin{equation*} 
  \bar C(H,T) = \frac{C(H,T)}{T\log_2 q}. 
\end{equation*} 
When we talk about the normalization of a coding rate, we mean 
to normalize by $T\log_2 q$. 
  
Achieving the capacity generally involves multiple usages of the channel.  
A block code for $\loc(H,T)$ is a subset of $(\ffield^{T\times M})^n$, 
the $n$th Cartesian power of $\ffield^{T\times M}$.  Here $n$ is the 
\emph{length} of the block code. Since the components of codewords are 
matrices, such a code is called a \emph{matrix code}.  The channel 
capacity of a LOC can be approached using a sequence of matrix codes 
with $n\rightarrow \infty$.

In the following subsection, we give the channel capacity and the 
capacity achieving inputs of three LOCs. These examples show that 
finding the channel capacity is problem-specific.  In general, it is 
not easy to accurately characterize the (noncoherent) capacity of a 
LOC. Since an input distribution contains $q^{TM}$ probability masses, 
a general method to maximize a mutual information, e.g., the 
Blahut-Arimoto algorithm, has time complexity $\bigO(q^{TM})$. 
Moreover, the distribution of the transformation matrix is difficult 
to obtain in applications like random linear network coding. 
Therefore, our goal is to find an efficient method to 
approach the capacity of LOCs with limited channel statistics.

\subsection{Examples of Linear Operator Channels} 
\label{sec:examples} 
 
\subsubsection{$Z$-Channel} 
 
A $Z$-channel with crossover probability $p$ is a 
binary-input-binary-output channel that flips the input bit 
$1$ with probability $p$, but maps input bit $0$ to $0$ with 
probability $1$. A $Z$-channel is a LOC over binary field 
given by 
\begin{equation*} 
  y = xh, 
\end{equation*} 
where $\Pr\{h=0\}=p$. 
We know the capacity of a $Z$-channel is $C(h,1)= \log_2 
\left(1+(1-p) p^{p/(1-p)}\right)$, which is achieved by 
\begin{equation*} 
  p_x(0)= \frac{1-p^{1/(1-p)}}{1+(1-p) p^{p/(1-p)}}. 
\end{equation*} 
 
\subsubsection{Full Rank Transformation Matrix} 
 
Let $H_{\text{full}}$ be the random matrix  
uniformly distributed over $\Fr(\ffield^{M\times N})$, 
$M\leq N$. For $\loc(H_{\text{full}},T)$, 
\begin{equation*} 
  P_{Y|X}(\bY|\bX) = \left\{ \begin{array}{ll} 
    \frac{1}{\cmat{N}{\rank(\bX)}} &  
    \lspan{\bY} = 
    \lspan{\bX} \\ 0 & \ow. \end{array}\right.  
\end{equation*} 
This kind of transformation matrix with $M=N$ has been 
studied in \cite{silva08c}. Let $M^*=\min\{M,T\}$. 
We know  
\begin{equation*} 
  C(H_{\text{full}},T) = \log_2 \sum_{r\leq M^*} \gco{T}{r}, 
\end{equation*} 
where $\sum_{r\leq M^*} \gco{T}{r} = |\Pj(M^*,\ffield^T)|$. 
Any input $p_X$ satisfying 
\begin{equation*} 
  p_{\lspan{X}}(U) = \frac{1}{|\Pj(M^*,\ffield^T)|},\quad 
  \forall U\in \Pj(M^*,\ffield^T), 
\end{equation*} 
is capacity achieving. In other words, this capacity is 
achieved by using each subspace in $\Pj(M^*,\ffield^T)$ 
uniformly.

\subsubsection{Purely Random Transformation Matrix} 
 
Recall that a random 
matrix is called \emph{purely random} if it contains 
uniformly independent components. Consider $\loc(H_{\text{pure}},T)$ 
with purely random $H_{\text{pure}}$ and dimension $M\times N$. We have 
\begin{equation*} 
  P_{Y|X}(\bY|\bX) = \left\{ \begin{array}{ll} 
    q^{-N\rank(\bX)} &  
    \lspan{\bY}\subseteq 
    \lspan{\bX} \\ 0 & \ow. \end{array}\right.  
\end{equation*} 
Such channels were studied in \cite{siav08, jafari09}, where 
the capacity formulas, involving big-O notations, are 
obtained for different cases.  
We will give an exact formula 
for sufficiently large $T$, 
\begin{equation*} 
  C(H_{\text{pure}},T) = 
  \E\left[\log_2\frac{\cmat{T}{\rank(H)}}{\cmat{M}{\rank(H)}}\right]. 
\end{equation*} 
This capacity is achieved by an input $p_X$ with 
\begin{equation*} 
  p_X(\bX)= \left\{ \begin{array}{ll} 1/\cmat{T}{M} & 
      \rank(\bX)=M \\ 0 & \ow. \end{array} \right.  
\end{equation*} 
In other words, this capacity is achieved by using all the 
full rank $T\times M$ matrices with equal probability.

\section{Channel Training} 
\label{sec:ct} 
 
In noncoherent transmission, the 
CI is not available in either the transmitter or the receiver. 
But we can deliver the CI to the receiver using a simple 
channel training technique.  When $T\geq M$, we can transmit an 
identity $M\times M$ matrix as a submatrix of $X$ to recover 
$H$ at the receiver. For example, if 
\begin{equation*} 
  X = \begin{bmatrix} \mathbf{I} \\ X' \end{bmatrix}, 
\end{equation*} 
then 
\begin{equation*} 
  Y = X H = \begin{bmatrix} H \\ X'H \end{bmatrix}. 
\end{equation*} 
The first $M$ rows of $Y$ gives the instance of $H$.  Thus the last 
$T-M$ rows of $Y$ can be decoded with the CI.  Let $C_{\text{CT}}$ be 
the maximum achievable rate of such a scheme, and $\bar C_{\text{CT}}$ 
be its normalization. 
 
\begin{proposition} 
  For $\loc(H,T)$ with dimension $M\times N$ and $T\geq M$, $\bar C_{\text{CT}}= (1-M/T)\E[\rank(H)]$. 
\end{proposition} 
\begin{IEEEproof} 
  Let $\tilde X$ be a random matrix over $\ffield^{(T-M)\times M}$ and 
  let $\tilde Y=\tilde XH$.  If the input of $\loc(H,T)$ is $X 
  = \begin{bmatrix} \mathbf{I} \\ \tilde X \end{bmatrix}$, the output 
  is 
  $Y = \begin{bmatrix} \mathbf{I} \\ \tilde X \end{bmatrix} H = \begin{bmatrix} H \\ \tilde Y \end{bmatrix}$. 
Thus,  
\begin{align*} 
  \bar C_{\text{CT}} & = \max_{p_{X}} I(X;Y)/(T\log_2q) \\  
   &  = \max_{p_{\tilde X}} I(\tilde X;\tilde YH)/(T\log_2q). 
\end{align*} 
Since $\tilde X$ and $H$ are independent, we have 
\begin{align*} 
  I(\tilde X;\tilde YH)
  & = I(\tilde X;\tilde Y|H)\\ 
  & = H(\tilde Y|H)\\ 
  & \leq \E[\rank(H)] (T-M)\log_2 q, 
\end{align*} 
where the equality is achieved by $\tilde X$ with uniformly independent 
components. 
\end{IEEEproof} 
 
\textbf{Remark:} In this formula of $\bar C_{\text{CT}}(H,T)$, $M/T$ 
is just the ratio of the overhead used in channel training. 
 
\begin{corollary}\label{lemma:bank} 
  $(1-M/T)\E[\rank(H)] \leq \bar 
  {C}(H,T) \leq  \E[\rank(H)]$. 
\end{corollary} 
\begin{IEEEproof} 
  It follows from ${C}_{\text{CT}}(H,T) \leq {C}(H,T) \leq 
  {C}_{\text{R-CI}}(H,T)$.  
\end{IEEEproof} 
 
The upper bound and the lower bound is asymptotically tight 
when $T$ is large. We will further improve the lower 
bound by showing that the inequality is strict.

Now we consider how to improve $\bar C_{\text{CT}}(H,T)$ by reducing 
the overhead ratio $M/T$. The method is to apply channel training to 
the new channel $\loc(GH,T)$ for a random matrix $G$ with dimension 
$r\times M$. See that $\bar C_{\text{CT}}(GH,T) = 
(1-r/T)\E[\rank(GH)]\leq (1-r/T)\E[\rank(H)]$. Thus, to achieve higher 
rate than $(1-M/T)\E[\rank(H)]$, we only need to consider $r<M$. We 
call this method \emph{extended channel training}. The maximum 
achievable rate of extended channel training is 
\begin{equation*} 
  \bar C_{\text{ECT}}(H,T) = \max_{r\leq M} \sup_{p_G:\text{The dimension of } G \text{ is } r\times M} (1-r/T)\E[\rank(GH)]. 
\end{equation*}

\begin{theorem}\label{the:ect} 
  For $\loc(H,T)$ with dimension $M\times N$, we have  
  \begin{align*} 
    \bar C_{\text{ECT}}(H,T) \geq \max\left\{\max_{r< M} (1-r/T)\left(\sum_{k=0}^{r-1} p_{\rank(H)}(k) k \cmatt{r}{k} + r\sum_{k=r}^{M} p_{\rank(H)}(k) \cmatt{k}{r}\right), \bar C_{\text{CT}}(H,T)\right\}, 
\end{align*} 
and 
\begin{align*} 
    \bar C_{\text{ECT}}(H,T) \leq \max_{r\leq M} (1-r/T)\left(\sum_{k=0}^{r-1} p_{\rank(H)}(k) k  + r\sum_{k=r}^{M} p_{\rank(H)}(k) \right). 
\end{align*} 
\end{theorem} 
\begin{IEEEproof}
  We have that  
  \begin{align*} 
    \E[\rank(GH)] & = \sum_{s=0}^r s \Pr\{\rank(GH)=s\} \\ 
    & = \sum_{s=0}^r s \sum_{k=s}^M \Pr\{\rank(GH)=s|\rank(H) = k\} 
    p_{\rank(H)}(k) \\ & = \sum_{k=0}^M p_{\rank(H)}(k) \sum_{s\leq 
      \min\{k,r\}} s \Pr\{\rank(GH)=s|\rank(H) = k\}. 
  \end{align*} 
  To prove the first inequality, we consider $G$ is purely random.  
  By Lemma~\ref{lm:NumMGivenRank}, $\Pr\{\rank(GH)=s|\rank(H) = k\} = \frac{\cmatt{k}{s}\cmatt{r}{s}}{\cmatt{s}{s} q^{(k-s)(r-s)}}$. 
  Then, 
  \begin{align*} 
    \E[\rank(GH)] & = \sum_{k=0}^M p_{\rank(H)}(k) \sum_{s\leq  \min\{k,r\}} s \frac{\cmatt{k}{s}\cmatt{r}{s}}{\cmatt{s}{s} q^{(k-s)(r-s)}} \\ &  
	= \sum_{k=0}^{r-1} p_{\rank(H)}(k) \sum_{s\leq k } s \frac{\cmatt{k}{s}\cmatt{r}{s}}{\cmatt{s}{s} q^{(k-s)(r-s)}} + \sum_{k=r}^{M} p_{\rank(H)}(k) \sum_{s\leq r } s \frac{\cmatt{k}{s}\cmatt{r}{s}}{\cmatt{s}{s} q^{(k-s)(r-s)}} \nonumber \\ &  
	> \sum_{k=0}^{r-1} p_{\rank(H)}(k) k \cmatt{r}{k} + \sum_{k=r}^{M} p_{\rank(H)}(k) r \cmatt{k}{r}. 
  \end{align*} 
  Therefore 
  \begin{align*} 
    \bar C_{\text{ECT}}(H,T) & = \max\left\{\max_{r< M} \sup_{p_G:\text{The dimension of } G \text{ is } r\times M} (1-r/T)\E[\rank(GH)], (1-M/T)\E[\rank(H)]\right\} \\ & = \max\left\{\max_{r< M} (1-r/T)\E[\rank(GH)]|_{G \text{ is purely random}}, (1-M/T)\E[\rank(GH)] \right\}\\ 
    & \geq \max\left\{\max_{r< M} (1-r/T) \left(\sum_{k=0}^{r-1} p_{\rank(H)}(k) k \cmatt{r}{k} + \sum_{k=r}^{M} p_{\rank(H)}(k) r \cmatt{k}{r}\right), (1-M/T)\E[\rank(GH)]\right \}. 
  \end{align*}

  To prove the second inequality, we see that 
  \begin{align*} 
    \E[\rank(GH)] & = \sum_{s=0}^r s \Pr\{\rank(GH)=s\} \\ 
    & = \sum_{s=0}^r s \sum_{k=s}^M \Pr\{\rank(GH)=s|\rank(G) = k\} 
    p_{\rank(G)}(k) \\ & = \sum_{k=0}^r p_{\rank(G)}(k) \sum_{s\leq 
      k} s \Pr\{\rank(GH)=s|\rank(G) = k\}. 
  \end{align*} 
  By 
  \begin{align*} 
    \sum_{s\leq k} s \Pr\{\rank(GH)=s|\rank(G) = k\} & = \sum_{r\leq k} \sum_{s \geq r} \Pr\{\rank(GH)=s|\rank(G) = k\} \\ & = \sum_{r\leq k} \Pr\{\rank(GH)\geq r|\rank(G) = k\} \\ & \leq \sum_{r\leq k} \Pr\{\rank(H)\geq r\} \\ & = \sum_{s<k} s p_H(s) + k \sum_{s\geq k} p_H(s), 
  \end{align*} 
  we have  
  \begin{align*} 
    \E[\rank(GH)] & = \sum_{k=0}^r p_{\rank(G)}(k) \left( \sum_{s<k} s p_H(s) + k \sum_{s\geq k} p_H(s) \right) \\ & \leq \max_{k\leq r} \left( \sum_{s<k} s p_H(s) + k \sum_{s\geq k} p_H(s) \right). 
  \end{align*} 
  The proof is completed. 
\end{IEEEproof}

\begin{corollary} 
Let $\bar C^{\text{upper}}_{\text{ECT}}(H,T)$ and $\bar C^{\text{lower}}_{\text{ECT}}(H,T)$ be the upper bound and the lower bound of $\bar C_{\text{ECT}}(H,T)$ in Theorem~\ref{the:ect}, respectively.  
When $T$ is sufficiently large,  
\begin{equation*} 
  \bar C^{\text{upper}}_{\text{ECT}}(H,T) - \bar C^{\text{lower}}_{\text{ECT}}(H,T) \leq \E[\rank(H)]\frac{M-r^*}{T}, 
\end{equation*} 
where $r^*=\max\{r:p_{\rank(H)}(r)>0\}$. This means that if $p_{\rank(H)}(M)>0$, $\bar C^{\text{upper}}_{\text{ECT}}(H,T) = \bar C^{\text{lower}}_{\text{ECT}}(H,T)$ when $T$ is sufficiently large.  
Fixing the rank distribution of $H$, we have  
\begin{equation*} 
  \lim_{q\rightarrow\infty} \bar C^{\text{lower}}_{\text{ECT}}(H,T) = \bar C^{\text{upper}}_{\text{ECT}}(H,T). 
\end{equation*} 
\end{corollary} 
\begin{proof} 
By the lower bound of $\bar C_{\text{ECT}}(H,T)$, we have $\bar C^{\text{lower}}_{\text{ECT}}(H,T) \geq (1-M/T)\E[\rank(H)]$. Let  
\begin{equation*} 
  a(r) = \sum_{s<r} s p_H(s) + r \sum_{s\geq r} p_H(s). 
\end{equation*} 
Since $a(r)=a(r^*)$ for all $r>r^*$, we only need to consider $r\leq r^*$ to find the upper bound of $\bar C_{\text{ECT}}(H,T)$, i.e.,  
\begin{equation*} 
  \bar C^{\text{upper}}_{\text{ECT}}(H,T) = \max_{r\leq r^*}(1-r/T) a(r). 
\end{equation*} 
Fix $r<r^*$. We know that $a(r)<a(r^*)$. Hence $(1-r/T)a(r)<(1-M/T)a(M)$ when $T\geq (r^*a(M)-ra(r))/(a(M-a(r))$. Therefore when $T\geq \max_{r<r^*} (r^*a(M)-ra(r))/(a(M)-a(r))$, $\bar C_{\text{ECT}}(H,T) = (1-r^*/T)\E[\rank(H)]$. 
 
The second part of this corollary follows from $\cmatt{m}{r}\rightarrow 1$ when $q\rightarrow \infty$. 
\end{proof} 
 
In Fig.~\ref{fig:bounds}, we  illustrate the bounds of $\bar C_{\text{ECT}}(H,T)$ and $\bar C_{\text{CT}}(H,T)$ over binary field.  
 
\begin{figure}[t] 
  \centering 
  \includegraphics[width=.8\textwidth]{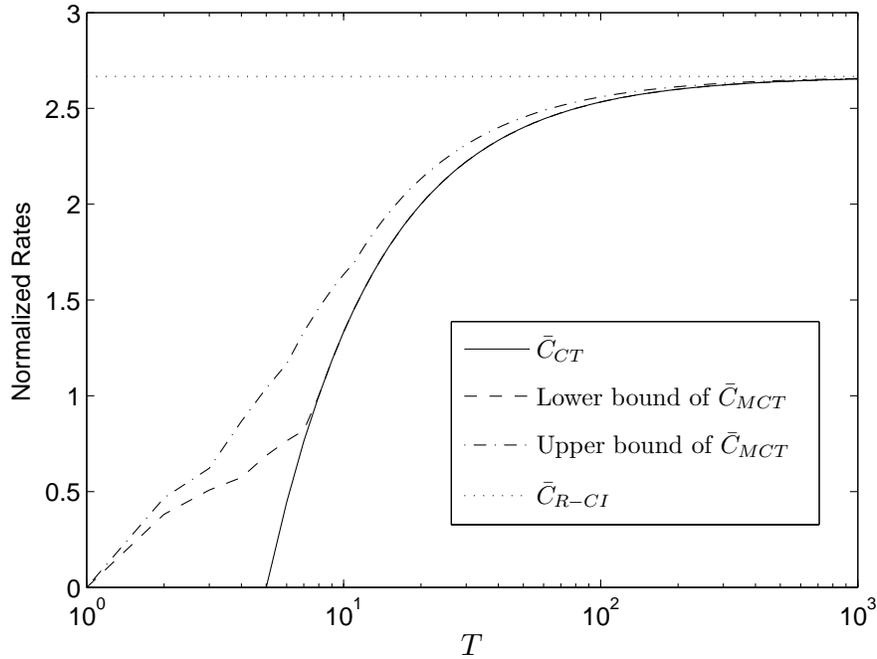} 
  \caption{Here we fix an $H$ with $M=5$ and $p_H(M)=0$ over binary field. We plot the lower and upper bounds of $\bar C_{\text{ECT}}(H,T)$ and $\bar C_{\text{CT}}(H,T)$ for 
    $T$ from $1$ to $1000$. Note that the bounds in this figure are only valid for 
    integer T and hence, the curves are not necessarily 
    smooth.} 
  \label{fig:bounds} 
\end{figure}

\section{Symmetric Property and Optimal Input Distributions} 
\label{sec:symm}

Here we introduce an intrinsic symmetric property of  
LOCs and show that this property is helpful 
to find an optimal input distribution of LOCs.  
 
\subsection{Random Variables and Markov Chains Related to LOCs}

We introduce several RVs related to LOCs, which are used 
extensively in this paper.   
Let $X$ be a RV over $\ffield^{t\times m}$. We denote by  
$\lspan{X}$ as a RV over $\Pj(\ffield^t)$ with 
\begin{equation} \label{eq:spaced} 
  p_{\lspan{X}}(U) = \Pr\{\lspan{X}=U\} = \sum_{\bX\in \ffield^{t\times m}:\lspan{\bX} = U} p_{X}(\bX). 
\end{equation} 
Denote $X^\tr$ as a RV over $\ffield^{m\times t}$ with 
$p_{X^\tr}(\bX^\tr)=p_{X}(\bX)$. 
Combining the above notations, $\lspan{X^\tr}$ is a RV over 
$\Pj(\ffield^m)$ with  
\begin{equation*} 
  p_{\lspan{X^\tr}}(V) = \sum_{\bX\in \ffield^{t\times m}:\lspan{\bX^\tr} 
  = V} p_{X}(\bX). 
\end{equation*} 
Furthermore, denote $\rank(X)$ as a RV with 
\begin{equation}\label{eq:rankd} 
  p_{\rank(X)}(r) = \sum_{\bX:\rank(\bX)=r}p_{X}(\bX). 
\end{equation} 
It is easy to see that $\rank(X)$ is a deterministic 
function of $\lspan{X}$ ($\lspan{X^\tr}$), and $\lspan{X}$ 
($\lspan{X^\tr}$) is a deterministic function of $X$. 
 
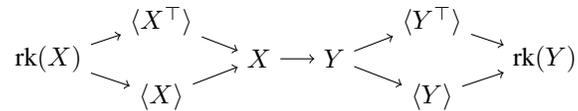
\begin{figure} 
  \centering 
  \begin{tikzpicture}[scale=0.5] 
    \node (X) at (-1,0) {$X$}; 
    \node (Y) at (1,0) {$Y$} edge[<-] (X); 
    \node (Y1) at (3.6,1) {$\lspan{Y^\tr}$} edge[<-] (Y); 
    \node (Y2) at (3.6,-1) {$\lspan{Y}$} edge[<-] (Y); 
    \node (X1) at (-3.6,1) {$\lspan{X^\tr}$} edge[->] (X); 
    \node (X2) at (-3.6,-1) {$\lspan{X}$} edge[->] (X); 
    \node at (6.6,0) {$\rank(Y)$} edge[<-] (Y1) edge[<-] 
    (Y2); 
    \node at (-6.6,0) {$\rank(X)$} edge[->] (X1) edge[->] (X2); 
  \end{tikzpicture} 
  \caption{Random variables and Markov chains related to $\loc(H,T)$.} 
  \label{fig:markov} 
\end{figure}

Now we consider $\loc(H,T)$ with dimension $M\times N$. Applying above 
definitions on the input $X$ and the output $Y$, we obtain the RVs 
shown in Fig.~\ref{fig:markov}.  These RVs are given as the nodes of a 
directed graph.  All the RVs in a directed path forms a Markov chain. 
For example, $\rank(X) \rightarrow \lspan{X} \rightarrow X \rightarrow 
Y \rightarrow \lspan{Y} \rightarrow \rank(Y)$ forms a Markov chain. 
Let $r$, $U$, $\bX$, $\bY$, $V$ and $s$ be the instances of 
$\rank(X)$, $\lspan{X}$, $X$, $Y$, $\lspan{Y}$ and $\rank(Y)$, 
respectively.  To verify this Markov chain, we only need to check the 
deterministic relations between these RVs: 
\begin{equation*} 
  p(r,U,\bX,\bY,V,s) 
  = \left\{\begin{array}{ll} p(\bX,\bY) &\text{if}\ 
      \lspan{\bX}=U, \dim(U)=r, \\  &\quad \lspan{\bY}=V, 
      \dim(V)=s, \\ 0 & \ow,\end{array}\right. 
\end{equation*} 
\begin{equation*} 
  p_{\rank(X)\lspan{X}}(r,U) 
  = \left\{\begin{array}{ll} p_{\lspan{X}}(U) &\text{if}\ 
      \dim(U)=r, \\ 0 & \ow,\end{array}\right. 
\end{equation*} 
and 
\begin{equation*} 
  p_{\lspan{Y}\rank(Y)}(V,s) 
  = \left\{\begin{array}{ll} p_{\lspan{Y}}(V) &\text{if}\ 
      \dim(V)=s, \\ 0 & \ow.\end{array}\right. 
\end{equation*} 
Using the above relations, we are ready to see 
\begin{align*} 
 \lefteqn{p(r,U,\bX,\bY,V,s)p(U)p(\bX)p(\bY)p(V)} \\ & 
  =p(r,U)p(U,\bX)p(\bX,\bY)p(\bY,V)p(V,s), 
\end{align*} 
which matches an alternative definition of Markov chain 
given in \cite[\S 2.1]{yeung08}. 
Other Markov chains shown in Fig.~\ref{fig:markov} can be verified accordingly.

\subsection{A Symmetric Property} 
 
The next proposition states a symmetric property of LOCs. Even though 
its proof is straightforward, this proposition plays a fundamental 
role in this paper. We say a matrix is full column (row) rank if its 
rank is equal to its number of columns (rows). 
 
\begin{proposition}\label{the:symm} 
  Consider $\loc(H,T)$.  
  For any matrix $\mathbf B$ with $T$ rows and full column rank,   
  \begin{equation*} 
    P_{Y|X}(\mathbf{BE}|\mathbf{BD}) = \Pr\{\mathbf{D} H = \mathbf{E}\}. 
  \end{equation*} 
\end{proposition} 
\begin{IEEEproof} 
  We know 
  \begin{align*} 
    P_{Y|X}(\mathbf{BE}|\mathbf{BD})  & = \Pr\{\mathbf{BD} H 
    = \mathbf{BE}\} \\ & = \Pr\{\mathbf{D} H = \mathbf{E}\}, 
  \end{align*} 
  where the last equality follows because $\mathbf B$ is 
  full column rank. 
\end{IEEEproof} 
 
Let $\mathbf B$ be a $t\times r$ matrix 
with rank $r$. For a  $t\times m$ matrix $\mathbf A$ with $\lspan{\mathbf 
  A}\subset \lspan{\mathbf B}$, define $\mathbf A/\mathbf B$ be the 
matrix such that $\mathbf A = \mathbf B(\mathbf A/\mathbf B)$.  The 
notation ``$/$'' is well defined because i) there must exists $\mathbf 
C$ such that $\mathbf A = \mathbf{BC}$ since $\lspan{\mathbf A}\subset 
\lspan{\mathbf B}$ and ii) such $\mathbf C$ is unique since $\mathbf 
B$ is full column rank.

Let $\bX$ and $\bY$ be the input and output matrices of a LOC, 
respectively, with $\lspan{\bY}\leq \lspan{\bX}$. Fix a full column 
rank matrix $\mathbf B$ with $\lspan{\bX} = \lspan{\mathbf B}$. 
Prop.~\ref{the:symm} tells that 
\begin{equation}\label{eq:symm} 
  P_{Y|X}(\bY|\bX) = \Pr\{ (\bX/\mathbf B) H = \bY/\mathbf B \}. 
\end{equation} 
The dimension of $\bX/\mathbf B$ is $\rank(\bX)\times M$ and the 
dimension of $\bY/\mathbf B$ is $\rank(\bX)\times N$.  This means that 
the transition probability $P_{Y|X}$ does not depends on the inaction 
period $T$. See examples in \S\ref{sec:examples}. 
In the following, we give two useful forms of this symmetric property.

\begin{corollary}\label{lemma:cond} 
  Let $\bX$ be an input matrix of $\loc(H,T)$. Then, 
  \begin{equation*} 
    P_{\rank(Y)|X}(s|\bX) = P_{\rank(Y)|\lspan{X^\tr}}(s|\lspan{\bX^\tr}) = 
    \Pr\{\rank(\mathbf{D} H)=s\}, 
  \end{equation*} 
  where $\mathbf D$ is any $\rank(\bX)\times M$ matrix with $\lspan{\mathbf 
    D^\tr} = \lspan{\bX^\tr}$. 
\end{corollary} 
\begin{IEEEproof} 
  Fix a $\rank(\bX)\times M$ matrix $\mathbf D$ with $\lspan{\bX^\tr} 
  = \lspan{\mathbf D^\tr}$. Let $\mathbf B^\tr = \bX^\tr/\mathbf 
  D^\tr$.  We know $\mathbf B$ is full column rank. 
  Since $X\rightarrow Y\rightarrow \rank(Y)$ forms a Markov chain, 
\begin{align} 
  P_{\rank(Y)|X}(s|\bX) & = \sum_{\bY} P_{\rank(Y)|Y}(s|\bY) 
  P_{Y|X}(\bY|\bX) \nonumber\\ 
  & = \sum_{\bY:\rank(\bY)=s} P_{Y|X}(\bY|\bX)\nonumber \\ 
  & = \sum_{\bY:\rank(\bY)=s} \Pr\{ \mathbf D H = \bY/\mathbf B \} \label{eq:fz} \\ 
  & = \sum_{\mathbf E:\rank(\mathbf E)=s} \Pr\{ \mathbf D H = \mathbf E \} \nonumber \\ 
  & = \Pr\{\rank(\mathbf{D} H)=s\},\nonumber 
\end{align} 
where \eqref{eq:fz} follows from \eqref{eq:symm}. 
 
Let $\tilde U= \lspan{\bX^\tr}$. 
By the Markov 
chain $\lspan{X^\tr} \rightarrow X \rightarrow \rank(Y)$,  
\begin{align} 
  \lefteqn{P_{\rank(Y)|\lspan{X^\tr}}(s|\tilde U )}\nonumber 
  \\ & = \sum_{\bX': \lspan{\bX'^\tr} = \tilde U} 
  P_{\rank(Y)|X}(s|\bX') P_{X|\lspan{X^\tr}}(\bX'|\tilde U) \nonumber\\ 
  & = \Pr\{\rank(\mathbf{D} H)=s\} \sum_{\bX': \lspan{\bX'^\tr} = \tilde U} 
   P_{X|\lspan{X^\tr}}(\bX'|\tilde U) \nonumber\\ 
  & = \Pr\{\rank(\mathbf{D} H)=s\}. \nonumber 
\end{align} 
The proof is completed. 
\end{IEEEproof}

\begin{corollary}\label{prop:symm} 
   Consider $\loc(H,T)$. For any $\Phi\in \Fr(\ffield^{T\times T})$, 
    \begin{equation}\label{eq:sb} 
      P_{Y|X}(\Phi \bY |\Phi \bX) = P_{Y|X}(\bY|\bX). 
    \end{equation} 
\end{corollary} 
\begin{IEEEproof} 
  This is a special cases of Prop.~\ref{the:symm}. 
\end{IEEEproof}

\subsection[alpha-type Input Distribution]{$\alpha$-type 
  Input Distributions} 
 
For a DMC, a capacity achieving input is also 
referred to as an optimal input. 
It is well known that the channel capacity of a symmetric channel is 
achieved by the symmetric input distribution \cite{gallager}.  
Even though in general LOCs are not symmetric channels, the symmetric 
property we have shown is still helpful to find an optimal input.

\begin{definition} 
   A PMF $p$ over $\ffield^{T\times M}$ is \emph{$\alpha$-type} 
  if $p(\bX)=p(\bX')$ for all 
  $\bX,\bX'\in\ffield^{T\times M}$ with $\lspan{\bX^\tr} = 
  \lspan{\bX'^\tr}$. 
\end{definition} 
 
For example, the input distribution 
\begin{equation*} 
  p_X(\bX)= \left\{ \begin{array}{ll} 1/\cmat{T}{M} & 
      \rank(\bX)=M \\ 0 & \ow \end{array} \right.  
\end{equation*} 
is \emph{the} $\alpha$-type input with $p_{\rank(X)}(M)=1$.

\begin{lemma}\label{lm:rw} 
A function $p:\ffield^{T\times M}\rightarrow 
\mathbb{R}$ is an $\alpha$-type PMF if and only if it can be written as 
\begin{equation}\label{eq:q} 
  p(\bX) =  Q(\lspan{\bX^\tr})/\cmat{T}{\rank(\bX)} 
\end{equation} 
for certain PMF $Q$ over $\Pj(\min\{M,T\},\ffield^M)$. 
\end{lemma} 
\begin{IEEEproof} 
  Assume $p$ is an $\alpha$-type input.  Define $Q: 
  \Pj(\min\{M,T\},\ffield^M)\rightarrow \mathbb{R}$ as 
  \begin{equation*} 
    Q(\tilde U)=\sum_{\bX'\in \ffield^{T\times M}:\lspan{\bX'^\tr}=\tilde U} p(\bX'). 
  \end{equation*} 
  For $\bX\in \ffield^{T\times M}$, 
  \begin{align*} 
    Q(\lspan{\bX^\tr}) & = \sum_{\bX'\in \ffield^{T\times M}:\lspan{\bX'^\tr}= \lspan{\bX^\tr}} p(\bX') \\ & = p(\bX)\sum_{\bX'\in \ffield^{T\times M}:\lspan{\bX'^\tr}= \lspan{\bX^\tr}}1 \\ & = p(\bX) \cmat{T}{\rank(\bX)}, 
  \end{align*} 
where the last equality follows from Lemma~\ref{lm:ia}. 
This proves the necessary condition. 
 
Now we prove the sufficient condition. 
Let $Q$ be a PMF over $\Pj(\min\{M,T\},\ffield^M)$. 
Define a function $p:\ffield^{T\times M}\rightarrow 
\mathbb{R}$ as  
\begin{equation*} 
  p(\bX) =  Q(\lspan{\bX^\tr})/\cmat{T}{\rank(\bX)}. 
\end{equation*} 
We can check that for $\bX,\bX'\in \ffield^{T\times M}$ with 
$\lspan{\bX^\tr} = \lspan{\bX'^\tr}$,  
\begin{align*} 
  p(\bX) &  = Q(\lspan{\bX^\tr})/\cmat{T}{\rank(\bX)} \\ 
  & = Q(\lspan{\bX'^\tr})/\cmat{T}{\rank(\bX)} \\ 
  & = p(\bX'),  
\end{align*} 
and 
\begin{align*} 
  \sum_{\bX} p(\bX) & = \sum_{\tilde U\in \Pj(\ffield^M)} 
  \sum_{\bX:\lspan{\bX^\tr}=\tilde U} 
  Q(\tilde U)/\cmat{T}{\dim(\tilde U)} \\ 
  & = \sum_{\tilde U\in \Pj(\ffield^M)} Q(\tilde U)/\cmat{T}{\dim(\tilde U)}\sum_{\bX:\lspan{\bX^\tr}=\tilde U} 1  \\ 
  & = \sum_{\tilde U\in \Pj(\ffield^M)} Q(\tilde U) \\ 
  & = 1. 
\end{align*} 
Thus $p$ is an $\alpha$-type PMF.  
\end{IEEEproof} 
 
The following proposition tells that we can only consider 
$\alpha$-type inputs to study the capacity of LOCs. 
 
\begin{theorem}\label{the:diq} 
  For a LOC there exists an $\alpha$-type input that 
  maximizes $I(X;Y)$. 
\end{theorem} 
\begin{IEEEproof} 
  This proposition is proved using Cor.~\ref{prop:symm} and 
  the concavity of mutual information as a function of input distribution. 
  See \S\ref{sec:rank2} for details. 
\end{IEEEproof}

Let $M^* = \min\{T,M\}$. Theorem~\ref{the:diq} narrows down the range to 
find an optimal input.  To determine a PMF over 
$\Pj(M^*,\ffield^{M})$, we have $|\Pj(M^*,\ffield^{M})|-1$ parameters to determine. We know $|\Pj(M^*,\ffield^{M})|-1 < q^{M^2/2 + 
  \log_qM +c}$, where $c<1.8$ is a constant (see 
Lemma~\ref{lemma:pjn}). But to determine a PMF over $\ffield^{T\times 
  M}$, we have to fix $q^{TM}-1$ parameters.  It is clear that $q^{M^2/2 + 
  \log_qM +c} / (q^{TM}-1) \rightarrow 0$ when $T\rightarrow \infty$, 
or when $T>M/2+1/e+c$ and $q\rightarrow \infty$.  Thus, using 
$\alpha$-type inputs can significantly reduce the complexity to find 
an optimal input distribution when i) $T$ is large or ii) 
$T>M/2+1/e+c$ and $q$ is large.

\subsection{Proof of Theorem~\ref{the:diq}} 
\label{sec:rank2}

\begin{lemma}\label{lemma:mc1} 
  Let $p_X$ be an input distribution of $\loc(H,T)$ with 
  dimension $M\times N$. Define $p'_X:\ffield^{T\times 
    M}\rightarrow \mathbb{R}$ as 
  $p'_X(\bX)=p_X(\Phi\bX)$, where $\Phi\in \Fr(\ffield^{T\times T})$. 
  We have, i) $p'_X$ is a PMF, ii) $I(X;Y)|_{p_X} = I(X;Y)|_{p'_X}$ 
  and iii) 
  $I(\lspan{X};\lspan{Y})|_{p_X} = I(\lspan{X};\lspan{Y})|_{p'_X}$. 
\end{lemma} 
\begin{IEEEproof}
  First $p'_X$ is a PMF because $0\leq p'_X(\bX) =p(\Phi\bX) \leq 1$ 
  and 
 \begin{align*} 
   \sum_{\bX\in \ffield^{T\times M}} p'_X(\bX) & = 
   \sum_{\bX\in \ffield^{T\times M}} p(\Phi \bX) \\ 
   & = \sum_{\bX\in \Phi \ffield^{T\times M}} p(\bX) \\ 
   & = \sum_{\bX\in \ffield^{T\times M}} p(\bX) \\ 
   & = 1. 
 \end{align*} 
 Let $p_Y$ and $p_Y'$ be the PMF of $Y$ when the inputs 
 are $p_X$ and $p'_X$, respectively. We have 
 \begin{align*} 
  p'_Y(\bY) & = \sum_{\bX\in \ffield^{T\times M}} p'_X(\bX) P_{Y|X}(\bY|\bX) \nonumber \\ 
	& \locequa{=}  \sum_{\bX\in \ffield^{T\times M}}p(\Phi\bX) P_{Y|X}(\Phi\bY|\Phi\bX)  \\ 
	& \locequa{=} \sum_{\bX'\in \ffield^{T\times M}}p(\bX') P_{Y|X}(\Phi\bY|\bX') \\ 
	& = p_Y(\Phi\bY). \nonumber  
 \end{align*} 
 where (a) follows from Cor.~\ref{prop:symm} 
 and $p'_X(\bX)=p_X(\Phi\bX)$,  
 and (b) follows by letting $\bX'=\Phi\bX$. Therefore, 
 \begin{align*} 
  I(X;Y)|_{p'_X} & = \sum_{\bX}p'_X(\bX) \sum_{\bY} P(\bY|\bX) 
  \log_2 \frac{P(\bY|\bX)}{p'_Y(\bY)} \nonumber \\ 
	& \locequa{=}  \sum_{\bX}p(\Phi\bX) \sum_{\bY} P(\Phi\bY|\Phi\bX) \log_2 \frac{P(\Phi\bY|\Phi\bX)}{p(\Phi\bY)} \\  
	& = \sum_{\bX'}p(\bX') \sum_{\bY'} P(\bY'|\bX') \log_2 
        \frac{P(\bY'|\bX')}{p(\bY')} \nonumber \\ 
	& = I(X;Y)|_{p_X}, \nonumber 
 \end{align*} 
 where (c) follows from Cor.~\ref{prop:symm}.

 The last equality in the lemma can be proved 
 similarly. First, 
 \begin{align*} 
   p'_{\lspan{X}}(U) & = \sum_{\bX: \lspan{\bX} = U}p'_X(\bX)\nonumber  \\ 
   & = \sum_{\bX: \lspan{\bX} = U}p_X(\Phi\bX) \nonumber \\ 
   & \locequa{=}  \sum_{\bX: \lspan{\bX} = \Phi U}p_X(\bX') \\ 
   & = p_{\lspan{X}}(\Phi U),\nonumber  
 \end{align*} 
 where (d) follows from Lemma~\ref{lm:ia}. 
 Let $P'_{\lspan{Y}|\lspan{X}}(V|U)$ be the transition 
 probability when the input is $p_X'$. 
 For $U\leq \ffield^T$ with $p_{\lspan{X}}(U)>0$, 
 \begin{align*} 
   \lefteqn{ P'_{\lspan{Y}|\lspan{X}}(V|U)} \\ & = \frac{\sum_{\bX, 
       \bY:\lspan{\bX} = U,\lspan{\bY} =V} P_{Y|X}(\bY |\bX) 
     p'_X(\bX)} { p'_{\lspan{X}}(U)} \\ 
   & = \frac{\sum_{\bX, 
       \bY:\lspan{\bX} = U,\lspan{\bY} =V} P_{Y|X}(\Phi\bY 
     |\Phi\bX) 
     p_X(\Phi\bX)} { p_{\lspan{X}}(\Phi U)} \\ 
   & = P_{\lspan{Y}|\lspan{X}}(\Phi V|\Phi U). 
 \end{align*} 
 Hence,  
 \begin{align*} 
   p'_{\lspan{Y}}(V) & = \sum_{U} 
   P'_{\lspan{Y}|\lspan{X}}(V|U)p'_{\lspan{X}}(U) \\ 
   & = \sum_{U} P_{\lspan{Y}|\lspan{X}}(\Phi V|\Phi U) 
   p_{\lspan{X}}(\Phi U) \\ 
   & = p_{\lspan{Y}}(\Phi V). 
 \end{align*} 
 Therefore, 
 \begin{align*} 
  \lefteqn{I(\lspan{X};\lspan{Y})|_{p'_X}} \\ & = 
  \sum_{U}p'_{\lspan{X}}(U) \sum_{V} P'(V|U)  
  \log_2 \frac{P'(V|U)}{P'_{\lspan{Y}}(V)}  \\ 
  & = \sum_{U}p_{\lspan{X}}(\Phi U) \sum_{V} P(\Phi V|\Phi U)  
  \log_2 \frac{P(\Phi V|\Phi U)}{p_{\lspan{Y}}(\Phi V)}  \\ 
  & = I(\lspan{X};\lspan{Y})|_{p_X}. 
 \end{align*} 
\end{IEEEproof}

\begin{IEEEproof}[Proof of Therem \ref{the:diq}] 
  Consider a LOC with inaction period $T$. 
  Let $p$ be an optimal input distribution for the 
  channel. For $\Phi\in \Fr(\ffield^{T\times T})$, define $p^{\Phi}$ as 
  $p^{\Phi}(\bX) = p(\Phi\bX)$.  By Lemma \ref{lemma:mc1}, 
  $p^{\Phi}(\bX)$ also achieves the capacity of the LOC. 
  Define $p^*$ as $$p^*(\bX) = 
 \frac{1}{|\Fr(\ffield^{T\times T})|}\sum_{\Phi\in \Fr(\ffield^{T\times T})} 
 p^{\Phi}(\bX).$$ By the concavity of the mutual information, 
 we know $p^*$ is also an optimal input for the channel.  
 
 Now we show that $p^*$ is $\alpha$-type. 
 Consider $\bX,\bX'\in\ffield^{T\times M}$ with 
 $\lspan{\bX^\tr} = \lspan{\bX'^\tr}$. By 
 Lemma~\ref{lemma:jingle}, there exists $\Phi_0\in 
 \Fr(\ffield^{T\times T})$ such that $\bX'=\Phi_0\bX$.  
 We have  
 \begin{align*} 
   p^*(\Phi_0\bX) & = \frac{1}{|\Fr(\ffield^{T\times T})|}\sum_{\Phi\in 
     \Fr(\ffield^{T\times T})} p^{\Phi}(\Phi_0\bX) \nonumber \\ & = 
   \frac{1}{|\Fr(\ffield^{T\times T})|}\sum_{\Phi\in \Fr(\ffield^{T\times T})} 
   p^{\Phi\Phi_0}(\bX) \nonumber \\ & = p^*(\bX) . 
 \end{align*} 
 where in the last equality we use $\Fr(\ffield^{T\times T}) = \Phi_0 
 \Fr(\ffield^{T\times T})$. 
\end{IEEEproof}

\section{Subspace Coding for Linear Operator Channels} 
\label{sec:subspace} 
 
Subspace coding was first proposed for noncoherent transmission of 
RLCNs. Here we generalize the idea to LOCs and study subspace coding 
from a general way. 
 
\subsection{Subspace Degradation of LOCs} 
 
In this section, we are interested in the Markov chain 
$\lspan{X}\rightarrow X \rightarrow Y \rightarrow \lspan{Y}$.  The 
transition probability from $X$ to $Y$ is given by (\ref{eq:trans}). 
The transition probability from $Y$ to $\lspan{Y}$ is deterministic: 
\begin{equation*} 
  P_{\lspan{Y}|Y}(V|\bY) = \left\{\begin{array}{ll}1 & 
      \lspan{\bY}=V \\ 0 & \ow. \end{array}\right. 
\end{equation*} 
Applying the property of Markov chain, we further know 
\begin{align*} 
  P_{\lspan{Y}|X}(V|\bX) & = \sum_{\bY} 
  P_{\lspan{Y}|Y}(V|\bY)P_{Y|X}(\bY|\bX) \\ 
  & = \sum_{\bY:\lspan{\bY}=V} P_{Y|X}(\bY|\bX). 
\end{align*} 
The transition probability $P_{X|\lspan{X}}$ is  
undetermined for a LOC.

\begin{definition} 
Consider $\loc(H,T)$ with transition probability $P_{Y|X}$. 
Given a transition probability $P_{X|\lspan{X}}$, we have a new 
channel law given by 
\begin{align} 
  P_{\lspan{Y}|\lspan{X}}(V|U) & =\sum_{\bX}P_{\lspan{Y}|X}(V|\bX) P_{X|\lspan{X}}(\bX|U) \nonumber \\ & = \sum_{\bX:\lspan{\bX}=U}\sum_{\bY:\lspan{\bY}=V} P_{Y|X}(\bY|\bX)P_{X|\lspan{X}}(\bX|U). \label{eq:decom} 
\end{align} 
This channel, called a subspace degradation of $\loc(H,T)$, takes 
subspaces as input and output. 
\end{definition}

A subspace degradation of $\loc(H,T)$ is identified by 
$P_{X|\lspan{X}}$. We take $\lspan{X}$ and $\lspan{Y}$ as the input 
and output of a subspace degradation, respectively.  The mutual 
information between $\lspan{X}$ and $\lspan{Y}$ can be written as a function of $p_{\lspan{X}}$ and $P_{\lspan{Y}|\lspan{X}}$, 
in which $P_{\lspan{Y}|\lspan{X}}$, given in \eqref{eq:decom}, is a 
function of $P_{X|\lspan{X}}(\bX|U)$. The capacity of a subspace 
degradation of a LOC is $\max_{p_{\lspan{X}}} I(\lspan{Y}, 
\lspan{X})$.  Therefore, the maximum achievable rate of subspace 
degradations of $\loc(H,T)$ is 
\begin{align*} 
  C_{\subs}(H,T)  = 
  \max_{p_{X|\lspan{X}}}\max_{p_{\lspan{X}}}I(\lspan{X};\lspan{Y}). 
\end{align*} 
The rate $C_{\subs}(H,T)$ is achievable since 
$\max_{p_{\lspan{X}}}I(\lspan{X};\lspan{Y})$ is 
achievable for any given $p_{X|\lspan{X}}$. 
 
\begin{lemma}\label{lm:ss} 
  For $\loc(H,T)$, $I(\lspan{X};\lspan{Y})$ is determined by $p_X$ and 
  we can write 
  \begin{equation*}
    C_{\subs}(H,T) = \max_{p_X}I(\lspan{X};\lspan{Y}). 
  \end{equation*} 
\end{lemma} 
\begin{IEEEproof}
  For a fixed LOC, we know that $I(\lspan{X};\lspan{Y})$ is determined 
  by $p_{\lspan{X}}$ and $P_{X|\lspan{X}}$. We show that 
  $p_{\lspan{X}}(U)$ and $P_{X|\lspan{X}}(\bX|U)$ appeared in 
  $I(\lspan{X};\lspan{Y})$ are determined by $p_X$.  First, we obtain 
  $p_{\lspan{X}}$ from $p_X$ as shown in (\ref{eq:spaced}).  Second, since 
\begin{align*} 
  P_{X|\lspan{X}}(\bX|U) p_{\lspan{X}}(U) & = 
  \Pr\{X=\bX,\lspan{X}=U\} \\ 
  & = \left\{\begin{array}{ll}p_{X}(\bX) & \lspan{\bX}=U \\ 0 & 
        \ow. \end{array}\right., 
\end{align*} 
we have 
\begin{equation}\label{eq:xxx} 
  P_{X|\lspan{X}}(\bX|U) = 
  \left\{\begin{array}{ll}\frac{p_{X}(\bX)}{p_{\lspan{X}}(U)} 
      & p_{\lspan{X}}(U)\neq 0,\ \lspan{\bX}=U \\ 0 & 
        \lspan{\bX}\neq U. \end{array}\right. 
\end{equation} 
That means, for $U$ with $p_{\lspan{X}}(U)>0$, 
$P_{X|\lspan{X}}(\bX|U)$ is determined by $p_X$.  Moreover, if 
$p_{\lspan{X}}(U)= 0$, $P_{X|\lspan{X}}(\bX|U)$ does not appear in 
$I(\lspan{X};\lspan{Y})$. Thus, $I(\lspan{X};\lspan{Y})$ can be 
regarded as a function with only one variable  $p_X$.  
This also implies that 
\begin{equation*} 
  C_{\subs}(H,T) \geq  \max_{p_X}I(\lspan{X};\lspan{Y}). 
\end{equation*} 
 
One the other hand, given $P_{X|\lspan{X}}$ and $p_{\lspan{X}}$, we have a PMF 
of $X$ given by  
\begin{equation*} 
  p_X(\bX)= p_{\lspan{X}}(\lspan{\bX})P_{X|\lspan{X}}(\bX|\lspan{\bX}), 
\end{equation*} 
which establishes that 
\begin{equation*} 
  C_{\subs}(H,T) \leq  \max_{p_X}I(\lspan{X};\lspan{Y}). 
\end{equation*} 
The proof is completed. 
\end{IEEEproof} 
 
In the following, we will treat $I(\lspan{X};\lspan{Y})$ as a function 
of $p_X$ for a given LOC. 
 
\begin{definition}\label{def:uniform} 
  $\loc(H,T)$ is \emph{uniform} if there exists a function 
  $\mu:\Pj(\ffield^T)\times \Pj(\ffield^T)\rightarrow 
  [0\ 1]$ such that  
\begin{equation*} 
   \Pr\{\bY=\bX H\} = \left\{ \begin{array}{ll} 
    \mu(\lspan{\bX},\lspan{\bY}) &  
    \lspan{\bY}\subseteq 
    \lspan{\bX} \\ 0 & \ow, \end{array}\right. \label{eq:uni0qq} 
\end{equation*} 
\end{definition} 
 
We can check that the three examples of LOCs in \S\ref{sec:examples} 
are all uniform. $C_{\subs}(H,T)$ gives a lower bound of 
$C(H,T)$. Moreover, this lower bound is tight for uniform LOCs. 
 
\begin{proposition}\label{lemma:uniform} 
  For a LOC, $I(X;Y)\geq I(\lspan{X};\lspan{Y})$ and the equality is 
  achieved by uniform LOCs. 
\end{proposition} 
\begin{IEEEproof} 
  See \S\ref{sec:lb}. 
\end{IEEEproof} 
 
\subsection{Subspace Coding} 
\label{sec:subspa} 
 
Since a subspace degradation of a LOC takes subspaces as input and 
output, the coding for this channel is called \emph{subspace coding}, 
which was first used by K\"otter and Kschischang for random linear 
network coding \cite{koetter08}. We give a generalized definition of 
subspace coding as follows.

Let $M^*= \min\{T, M\}$ and recall that $\Pj(M^*,\ffield^T)$ is the 
set of subspaces of $\ffield^T$ with dimension less than or 
equal to $M^*$. The $n$th Cartesian power of 
$\Pj(M^*,\ffield^T)$ is $\Pj^n(M^*,\ffield^T)$.  An $n$-block 
subspace code is a subset of $\Pj^n(M^*,\ffield^T)$. 
Recall that {the Grassmannian} $\Gr(r,\ffield^T)$ is the set of all 
$r$-dimensional subspaces of $\ffield^T$.  An 
$r$-dimensional (constant-dimensional) subspace code is a 
subset of $\Gr^n(r,\ffield^T)$, the $n$th Cartesian power of 
$\Gr(r,\ffield^T)$.

For $\loc(H,T)$, we can choose a transition probability 
$P_{X|\lspan{X}}$ and apply a subspace code to its subspace 
degradation with $P_{X|\lspan{X}}$.  In other word, we transmit $U\in 
\Pj(M^*,\ffield^T)$ through the LOC by randomly 
choosing a matrix $\bX$ according to the transition probability 
$P_{X|\lspan{X}}(\bX|U)$.  
The decoding of a subspace code only consider the subspace spanned by 
the channel output. So, for two reception $\bY$ and $\bY'$ with 
$\lspan{\bY}=\lspan{\bY'}$, a subspace code decoder treats them as the 
same. The maximum achievable rate of subspace  coding for $\loc(H,T)$ is given by $C_{\subs}(H,T)$.

\subsection{A Decomposition of Mutual Information} 
 
\begin{theorem}\label{the:diq2} 
  For a LOC there exists an $\alpha$-type input that 
  maximizes $I(\lspan{X};\lspan{Y})$. 
\end{theorem} 
\begin{IEEEproof} 
  This proposition can be proved similar to Theorem~\ref{the:diq} by 
  applying Lemma~\ref{lemma:mc1}. 
\end{IEEEproof} 
 
By Theorem~\ref{the:diq2}, we know  
\begin{equation*} 
  C_{\subs}(H,T) = 
  \max_{p_{X}:\alpha\text{-type}}I(\lspan{X};\lspan{Y}). 
\end{equation*} 
That is, we only need to consider $\alpha$-type inputs to find $C_{\subs}(H,T)$.  
 
For a random matrix $X$, recall that $\rank(X)$ is the RV 
representing the rank of $X$ (see (\ref{eq:rankd}) for the 
PMF). Similar to Lemma~\ref{lm:ss}, for a LOC 
 $I(\rank(X);\rank(Y))$ is determined by $p_X$ and $P_{Y|X}$. 
Define 
\begin{align} 
    {J}(\rank(X);\rank(Y)) & =\sum_{s,r}   
  p_{\rank(X)\rank(Y)}(r,s) 
  \log_2\frac{\cmat{T}{s}}{\cmat{r}{s}} \nonumber \\ 
  & = 
  \E\left[\log_2\frac{\cmat{T}{\rank(Y)}}{\cmat{\rank(X)}{\rank(Y)}}\right], \label{eq:112345d}   
\end{align} 
where $ p_{\rank(X)\rank(Y)}(r,s)$ can be derived using $p_X$ and $P_{Y|X}$.

\begin{lemma} \label{the:8glsie} 
  For a LOC with $\alpha$-type inputs,  
  \begin{equation} \label{eq:i8gaqif} 
    I(\lspan{X};\lspan{Y}) =  I(\rank(X);\rank(Y)) + 
    {J}(\rank(X);\rank(Y)).  
  \end{equation} 
\end{lemma} 
\begin{IEEEproof} 
  The proof is done by rewriting the formulation of 
  mutual information using the symmetric property and the 
  definition of $\alpha$-type inputs. See \S\ref{sec:lb} 
  for details. 
\end{IEEEproof} 
 
In (\ref{eq:i8gaqif}), $I(\rank(X);\rank(Y))$ is the mutual 
information of the ranks of transmitted and received matrices.  In 
other words, it is the rate transmitted using the matrix ranks.  The 
meaning of $J(\rank(X);\rank(Y))$ has an interpretation using set 
packing. 
The capacity contributed by $r$-dimensional transmissions and 
$s$-dimension receptions is $\log_2 
\frac{\cmat{T}{s}}{\cmat{r}{s}} = \log_2 {\gco{T}{s}}/{\gco{r}{s}}$,  
where $\gco{T}{s}$ is the total number of $s$-dimensional subspaces in 
$\ffield^T$, and $\gco{r}{s}$ is the total number of $s$-dimensional 
subspaces in an $r$-dimensional subspace. 
Treat an $s$-dimensional subspace in $\ffield^T$ as a set 
element.  
An $r$-dimension transmission can be regarded as a collection of 
$s$ dimensional subspaces that span it.  
Then, the \emph{maximum set packing} problem is looking for 
the maximum number of pairwise disjoint collections of $s$-dimensional 
subspaces that has cardinality $\gco{M}{r}$ and spans an 
$M$-dimensional subspace.

\subsection{Lower Bound of the Maximum Achievable Rate} 
 
Using Lemma~\ref{the:8glsie}, we  
derive two lower bounds of the maximum achievable rates of subspace 
coding that only depend on the rank distribution.

\begin{theorem}\label{lemma:low} 
  For  $\loc(H,T)$ with dimension $M\times N$ and $T\geq M$, 
  \begin{align} 
    \bar C_{\subs}(H,T) & \geq 
    \E\left[\log_2\frac{\cmat{T}{\rank(H)}}{\cmat{M}{\rank(H)}}\right]/(T\log_2 q)  
    \nonumber \\ & = (1-M/T)\E[\rank(H)]+ \epsilon(T,q), \label{eq:lou} 
  \end{align} 
  where 
  \begin{align*} 
    \epsilon(T,q)T\log_2q= \sum_{s} p_{\rank(H)}(s) \log_2 
  \frac{\cmatt{T}{s}}{\cmatt{M}{s}} <1.8. 
  \end{align*} 
  This lower bound is achieved by the $\alpha$-type input $p_X$ with $p_{\rank(X)}(M)=1$. 
\end{theorem} 
\begin{IEEEproof} 
  See \S\ref{sec:lb}. 
\end{IEEEproof} 
\textbf{Remark:} Note that this bound depends on the 
rank distribution of the transformation matrix.  
This lower bound is tight for certain LOCs 
with sufficiently large $T$ (see Theorem~\ref{the:mmm}).

The RHS of \eqref{eq:lou} implies that subspace coding can achieve higher rate than channel training.  
As a quick summary, we see 
\begin{equation}\label{eq:chanel} 
  (1-M/T) \E[\rank(H)] + \epsilon(T,q)\leq \bar C_{\subs}(H,T) 
  \leq \bar C(H,T) \leq \E[\rank(H)]. 
\end{equation} 
This lower bound is better than the one in 
Cor.~\ref{lemma:bank}. Furthermore, 
\begin{align*} 
  \bar C(H,T) - \bar C_{\subs}(H,T) & \leq \E[\rank(H)] - 
  (1-M/T) \E[\rank(H)] - \epsilon(T,q) \\ 
  & = M/T \E[\rank(H)] - \epsilon(T,q). 
\end{align*} 
 
\begin{figure}[t] 
  \centering 
  \includegraphics[width=.8\textwidth]{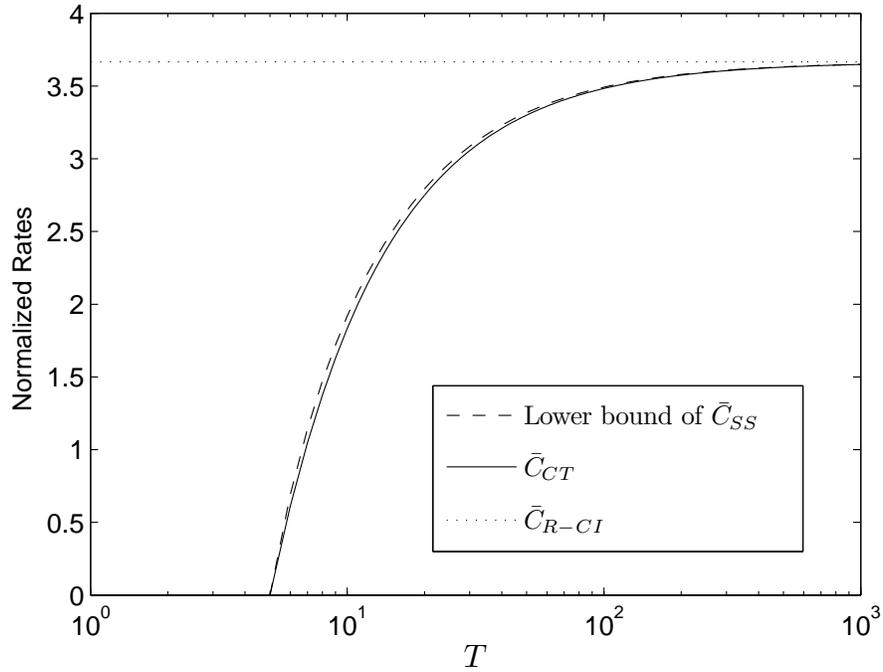} 
  \caption{Here we fix an $H$ with $M=5$ over binary field. We plot the lower bounds for 
    $T$ from $5$ to $1000$.} 
  \label{fig:lowerbound} 
\end{figure} 
 
The gap between the lower bound of $\bar C_{\subs}(H,T)$ and $\bar C_{\text{CT}}(H,T)$ is quit small, which is demonstrated in Fig.~\ref{fig:lowerbound}.  
Prop.~\ref{lemma:low}, however, is trivial for $T\leq M$. 
We can use the similar method in extended channel training to obtain a better lower bound. We can foresee that the improved lower bound of $\bar C_{\subs}(H,T)$ is close to $\bar C_{\text{ECT}}(H,T)$. We will not repeat the procedure here.

\subsection{Proofs} 
 \label{sec:lb}

\begin{IEEEproof}[Proof of Prop.~\ref{lemma:uniform}] 
   Let $\mathcal{U} = \Pj(\ffield^T)$. 
   We have 
   \begin{align*} 
    I(X;Y) & = \sum_{\bX, \bY} p_{X,Y}(\bX, \bY) \log_2 \frac{p_{X,Y}(\bX, \bY)}{p_X(\bX)p_{Y}(\bY)} \nonumber \\  
	  & = \sum_{V,U\in \mathcal{U}} \sum_{\substack{\bX,\bY:\\\lspan{\bX}=U,\lspan{\bY}=V}} p(\bX,\bY) \log_2 \frac{p(\bX,\bY)}{p_X(\bX)p_{Y}(\bY)} \nonumber \\  
	  & \locequa{\leq} \sum_{V, U\in \mathcal{U}} 
          p_{\lspan{X}\lspan{Y}}(U,V) \log_2 
          \frac{p_{\lspan{X}\lspan{Y}}(U,V)}{p_{\lspan{X}}(U)p_{\lspan{Y}}(V)} \\ 
          & = I(\lspan{X};\lspan{Y}), \nonumber 
   \end{align*} 
   where (\thelocequator) follows from the log-sum inequality. 
  To prove this proposition, we only need to show the 
  equality in (\thelocequator) holds for uniform LOCs.  
  We need to check that 
  $P_{Y|X}(\bY|\bX)/p_Y(\bY)$ is a constant for all $\bX$ 
  and $\bY$ with $\lspan{\bY}=V\leq \lspan{\bX}=U \leq 
  \ffield^T$.  
  Fix an input distribution $p_X$.  
  Since the LOC is uniform,  
  \begin{align*} 
    p_Y(\bY) & = \sum_{\bX:V\leq \lspan{\bX}} P_{Y|X}(\bY|\bX) p_X(\bX) \\ 
    & =  \sum_{U'\leq \ffield^T:V\leq 
      U'} \mu(U,V) \sum_{\bX:\lspan{\bX}=U' }  p_{X}(\bX) 
    \\& = \sum_{U'\leq \ffield^T:V\leq 
      U'} \mu(U',V) p_{\lspan{X}}(U'). 
  \end{align*} 
  Thus, 
   \begin{align*} 
     \frac{P_{Y|X}(\bY|\bX)}{p_Y(\bY)} & = \frac{\mu(U,V)}{\sum_{U'\leq \ffield^T:V\leq 
      U'} \mu(U',V) p_{\lspan{X}}(U')}.  
   \end{align*} 
   This verifies the equality in (\thelocequator) holding. 
\end{IEEEproof}

\begin{IEEEproof}[Proof  of Lemma \ref{the:8glsie}] 
Fix an $\alpha$-type input $p_X$. For $V\leq U\leq \ffield^T$ with 
$\dim(U)=r$ and $\dim(V)=s$, we first show  
\begin{equation}\label{eq:uv} 
  p_{\lspan{X}\lspan{Y}}(U,V) = \frac{p_{\rank(X)\rank(Y)}(r,s)}{\gco{T}{r}\gco{r}{s}}. 
\end{equation} 
We only need to show  that $p_{\lspan{X}\lspan{Y}}(U,V)= 
  p_{\lspan{X}\lspan{Y}}(U',V')$ for any $V\leq U\leq 
  \ffield^T$ and $V'\leq U'\leq \ffield^T$ with 
  $\dim(U)=\dim(U')$ and $\dim(V)=\dim(V')$, because if this 
  is true, 
  \begin{align*} 
   p_{\rank(X)\rank(Y)}(r,s) & = \sum_{\dim(U^*)=r,\dim(V^*)=s, V^*\subset U^*} p_{\lspan{X}\lspan{Y}}(U^*,V^*)\\ 
     & =p_{\lspan{X}\lspan{Y}}(U,V) \sum_{\dim(U^*)=r,\dim(V^*)=s, V^*\subset U^*}1 \\ 
     &= p_{\lspan{X}\lspan{Y}}(U,V) \gcos{T}{r} \gcos{r}{s}. 
  \end{align*}

  Let  
  \begin{equation*} 
  A(m,U) = \{\bX\in \ffield^{t\times m}:\lspan{\bX} = U\}. 
  \end{equation*} 
  By Lemma~\ref{lm:oop}, we can find $\Phi\in 
  \Fr(\ffield^{T\times T})$ such that $\Phi U=U'$ and $\Phi V = V'$. 
  Then, 
  \begin{align*} 
    p_{\lspan{X}\lspan{Y}}(U,V) 
    & = {\sum_{\bX\in A(M,U)} p_{X}(\bX) 
      \sum_{\bY\in A(N,V)} P_{Y|X}(\bY|\bX)} \nonumber \\ 
    & \locequa{=} {\sum_{\bX\in A(M,U)} p_{X}(\Phi\bX) 
      \sum_{\bY\in A(N,V)} P_{Y|X}(\Phi\bY|\Phi\bX)} \\ 
    & \locequa{=} {\sum_{\bX\in A(M,\Phi U)} p_{X}(\bX) 
      \sum_{\bY\in A(N,\Phi V)} P_{Y|X}(\bY|\bX)} \\ 
    & = p_{\lspan{X}\lspan{Y}}(\Phi U,\Phi V) \nonumber\\ 
    & = p_{\lspan{X}\lspan{Y}}(U',V').\nonumber 
  \end{align*} 
  (b) follows that $p_X$ is $\alpha$-type ($p_X(\bX)=p_X(\Phi\bX)$) 
  and $P_{Y|X}(\Phi\bY|\Phi\bX) = P_{Y|X}(\bY|\bX)$ follows from 
  Cor.~\ref{prop:symm}.  (c) follows from $A(m,\Phi U) = \Phi A(m,U)$ 
  (see Lemma~\ref{lm:ia}).  This proves \eqref{eq:uv}. 
 
  Applying the property of $\alpha$-type input, 
  \begin{align*} 
    p_{\lspan{X}}(U) & = \sum_{\bX\in A(M,U)} p_X(\bX)\nonumber \\ 
    & = \sum_{\bX\in A(M,U)} p_X(\Phi\bX)\nonumber \\ 
    & = \sum_{\bX\in \Phi A(M,U)} p_X(\bX)\nonumber \\ 
    & \locequa{=} \sum_{\bX\in A(M,U')} p_X(\bX) \\ 
    & = p_{\lspan{X}}(U') \nonumber 
  \end{align*} 
  where (\thelocequator) follows from Lemma~\ref{lm:ia}). 
Therefore, 
\begin{equation}\label{eq:u} 
  p_{\lspan{X}}(U) = \frac{p_{\rank(X)}(r)}{\gco{T}{r}}. 
\end{equation} 
Moreover, 
\begin{align} 
  p_{\lspan{Y}}(V) & = \sum_{U:V\subset U} 
   p_{\lspan{X}\lspan{Y}}(U,V)\nonumber\\   
  & = \sum_{r\geq s} \sum_{U:V\subset U, \dim(U)=r} 
   p_{\lspan{X}\lspan{Y}}(U,V)\nonumber\\  
   & = \sum_{r\geq s}  \frac{p_{\rank(X)\rank(Y)}(r,s)}{\gco{T}{r}\gco{r}{s}} \sum_{U:V\subset U,\dim(U)=r}1 \nonumber\\ 
   & \locequa{=} \sum_{r\geq s} \frac{p_{\rank(X)\rank(Y)}(r,s)}{\gco{T}{r}\gco{r}{s}} \gcos{T-s}{r-s} \nonumber\\ 
   & \locequa{=} \sum_{r\geq s} \frac{p_{\rank(X)\rank(Y)}(r,s)}{\gco{T}{r}\gco{r}{s}} \gcos{T}{r}\frac{\cmat{r}{s}}{\cmat{T}{s}} \nonumber\\ 
   & = \frac{p_{\rank(Y)}(s)}{\gco{T}{s}}, \label{eq:v} 
\end{align} 
where (e) and (f) follow from 
Lemma \ref{lemma:8balfa}.  Substituting \eqref{eq:uv}, \eqref{eq:u} 
and \eqref{eq:v} into $I(\lspan{X};\lspan{Y})$, we have 
  \begin{eqnarray*}  
  \lefteqn{I(\lspan{X};\lspan{Y}) } \\ & =  & \sum_{V\leq U} p_{\lspan{X}\lspan{Y}}(U,V) \log_2 
  \frac{p_{\lspan{X}\lspan{Y}}(U,V)}{p_{\lspan{X}}(U)p_{\lspan{Y}}(V)}\\ 
  & = & \sum_{s\leq r} \sum_{\substack{V\leq U, \dim(U)=r,\\ \dim(V)=s}} 
  p_{\lspan{X}\lspan{Y}}(U, V) \log_2  
  \frac{p_{\lspan{X}\lspan{Y}}(U, V)}{p_{\lspan{X}}(U)p_{\lspan{Y}}(V)}\\ 
  & = & \sum_{s\leq r} \sum_{\substack{ V\leq U, \dim(U)=r,\\ \dim(V)=s}} 
  p_{\lspan{X}\lspan{Y}}(U, V) \log_2  
  \frac{p_{\rank(X)\rank(Y)}(r,s)}{p_{\rank(X)}(r)p_{\rank(Y)}(s)} 
  \frac{\gco{T}{s}}{\gco{r}{s}}\\ 
  & = & \sum_{s\leq r} p_{\rank(X)\rank(Y)}(s,r) \log_2  
  \frac{p_{\rank(X)\rank(Y)}(r,s)}{p_{\rank(X)}(r)p_{\rank(Y)}(s)} 
  \frac{\gco{T}{s}}{\gco{r}{s}}  \\ 
  & = & I(\rank(X);\rank(Y)) + \sum_{s\leq r}  
  p_{\rank(X)\rank(Y)}(r,s) 
  \log_2\frac{\cmat{T}{s}}{\cmat{r}{s}}. 
  \end{eqnarray*} 
  This completes the proof. 
\end{IEEEproof} 
 
\begin{IEEEproof}[Proof of Theorem~\ref{lemma:low}] 
  Substituting the $\alpha$-type 
  input with $p_{\rank(X)}(M)=1$ in Lemma~\ref{the:8glsie}, we have 
  $I(\rank(X);\rank(Y)) = 0$ and $J(\rank(X);\rank(Y)) = \sum_{s} 
  P_{\rank(Y)|\rank(X)}(s|M) \log_2\frac{\cmat{T}{s}}{\cmat{M}{s}}$. 
  Given $\bX\in \ffield^{T\times M}$ with dimension $M$, 
  \begin{equation*} 
    P_{\rank(Y)|X}(s|\bX) = \Pr\{\rank(\bX H) = s\} = \Pr\{\rank(H) = s\}. 
  \end{equation*} 
  Thus, $P_{\rank(Y)|\rank(X)}(s|M) = \Pr\{\rank(H) = s\}$.   
  Using the definition in 
  (\ref{eq:speaker}), we can write 
  \begin{align*} 
    \log_2 \frac{\cmat{T}{s}}{\cmat{M}{s}} & = \log_2 
    \frac{\cmatt{T}{s}q^{Ts}}{\cmatt{M}{s}q^{Ms}}  \\ 
    & = (T-M)s \log_2 q + \log_2 \frac{\cmatt{T}{s}}{\cmatt{M}{s}}. 
  \end{align*} 
  Since $\cmatt{T}{s}<1$, 
  \begin{equation} 
    \label{eq:178} 
    \log_2 \frac{\cmatt{T}{s}}{\cmatt{M}{s}} < \log_2 
  \frac{1}{\cmatt{M}{s}} < 1.8, 
  \end{equation} 
  where the last inequality follows from Lemma~\ref{lm:eight}. 
  So 
  \begin{align*} 
    J(\rank(X);\rank(Y)) &  = \sum_{s} 
  p_{\rank(H)}(s) (T-M)s \log_2 q + \\ & \quad \sum_{s} 
  p_{\rank(H)}(s) \log_2 \frac{\cmatt{T}{s}}{\cmatt{M}{s}}  \\ 
  & = (T-M)\log_2 q \E[\rank(H)]+ \epsilon(T,q)T\log_2q, 
  \end{align*} 
  where $\epsilon(T,q) = \sum_{s} p_{\rank(H)}(s) \log_2 
  \frac{\cmatt{T}{s}}{\cmatt{M}{s}}/(T\log_2q) < 1.8/(T\log_2q)$.  The 
  proof is complete by $C_{\subs}(H,T) \geq J(\rank(X);\rank(Y))$. 
\end{IEEEproof}

\section{Optimal Inputs for Subspace Coding} 
\label{sec:input} 
 
In this section, we show that using constant-dimensional subspace 
coding is almost as good as the general (multi-dimensional) subspace 
coding.

\subsection{A Formulation of $\alpha$-type Inputs} 
 
\begin{lemma} 
  A function $p:\ffield^{T\times M}\rightarrow \mathbb{R}$ is an $\alpha$-type PMF if and only if it can be written as  
\begin{equation}\label{eq:q2} 
  p(\bX)= R(\rank(\bX)) 
  \frac{Q_{\rank(\bX)}(\lspan{\bX^\tr})}{\cmat{T}{\rank(\bX)}} 
\end{equation} 
where $Q_r(\cdot)$ is a PMF over $\Gr(r,\ffield^M)$ and 
$R(\cdot)$ be a PMF over $\{0,1,\cdots,M\}$.  
\end{lemma} 
\begin{IEEEproof} 
  If $p$ can be written as \eqref{eq:q2}, by Lemma~\ref{lm:rw}, $p$ is 
  an $\alpha$-type PMF. On the other hand, if $p$ is an $\alpha$-type 
  PMF, it can be written as \eqref{eq:q}.  
  Let  
\begin{equation*} 
  R(r) = \sum_{\tilde U:\rank(\tilde U)=r}Q(\tilde U). 
\end{equation*} 
For $r$ such that $R(r)>0$, let 
\begin{equation*} 
  Q_r(\tilde U) = \left\{ \begin{array}{ll} Q(\tilde U)/R(r) & \dim(\tilde U)=r \\ 0 & \ow. \end{array}\right. 
\end{equation*} 
For $r$ such that $R(r)=0$, let $Q_r(\cdot)$ be any PMF over 
$\Gr(r,\ffield^M)$.  Since $Q_{\dim(\tilde U)}(\tilde U) R(\dim(\tilde 
U))= Q(\tilde U)$, we see that $p$ can be written as \eqref{eq:q2}. 
\end{IEEEproof} 
 
When using the formulation in \eqref{eq:q2},  
$I(\rank(X);\rank(Y))$ and 
$J(\rank(X);\rank(Y))$ can be written as functions of $Q_r(\tilde U)$ and 
$R(r)$ as follows. 
Using the property of Markov chain,  
\begin{align} 
  \lefteqn{P_{\rank(Y)|\rank(X)}(s|r)} \nonumber \\ & = \sum_{\tilde U\in 
    \Gr(r,\ffield^M)} P_{\rank(Y)|\lspan{X^\tr}}(s|\tilde 
  U)P_{\lspan{X^\tr}|\rank(X)}(\tilde U|r) \nonumber \\ 
  & = \sum_{\tilde U\in 
    \Gr(r,\ffield^M)} P_{\rank(Y)|\lspan{X^\tr}}(s|\tilde 
  U)Q_r(\tilde U),\label{eq:fry} 
\end{align} 
in which $P_{\rank(Y)|\lspan{X^\tr}}(s|\tilde U)$, given 
in Coro.~\ref{lemma:cond}, is a function of $p_H$ and is not related 
to $Q_r(\tilde U)$ and $R(r)$. Thus, we can write 
\begin{equation}\label{eq:ind} 
  {I(\rank(X);\rank(Y))} =  
  \sum_r R(r) \sum_sP(s|r) \log_2 
  \frac{P(s|r)}{\sum_{r'}P(s|r')R(r')},  
\end{equation} 
in which $P(s|r) = P_{\rank(Y)|\rank(X)}(s|r)$ is given in (\ref{eq:fry}); 
and 
\begin{align*} 
  {J}(\rank(X);\rank(Y)) & =  
  \sum_r R(r) \sum_{\tilde U \in \Gr(r,\ffield^M)} 
   Q_r(\tilde U) g(\tilde U),  
\end{align*} 
in which 
\begin{equation}\label{eq:g} 
  g(\tilde U) \triangleq \sum_s P_{\rank(Y)|\lspan{X^\tr}}(s|\tilde 
  U) \log_2\frac{\cmat{T}{s}}{\cmat{\dim(\tilde U)}{s}}. 
\end{equation} 
Note that $g(\tilde U)$ only depends on the distribution of 
$H$, but does not depend on the input. 
Define 
$$\rank^*(H)=\max\{r:\Pr\{\rank(H)=r\}>0\},$$ i.e., the maximum nonzero 
rank of the transformation matrix. 
 
\begin{lemma}\label{lemma:food2} 
  Consider $\loc(H,T)$ with dimension $M\times N$ and $T\geq 
  M$. Fix an $\alpha$-type input.  
  For $\tilde V\leq 
  \ffield^M$ with $\dim(\tilde V)=r<\rank^*(H)$, 
  \begin{equation*} 
    g(\ffield^M) - g(\tilde V) > \Theta(T,r,H)\log_2 q,  
  \end{equation*} 
  where 
  \begin{align*} 
    \Theta(T,r,H) & = (T-M)(\rank^*(H)-r)  
       p_{\rank(H)}(\rank^*(H)) \\ & \qquad -  r(M-r) + \log_q \cmatt{r}{r}. 
  \end{align*} 
\end{lemma} 
\begin{IEEEproof} 
  See \S\ref{sec:p9}. 
\end{IEEEproof}

\subsection{Optimal Inputs for Subspace Coding} 
 
We treat $Q_r(\bX)$ and $R(r)$ as the variables to 
maximize $I(\lspan{X};\lspan{Y})$.   
By the KKT conditions, a set of necessary and sufficient 
conditions such that an 
$\alpha$-type input with variables $Q_r(\bX)$ and $R(r)$ to achieve 
$C_{\subs}(H,T)$ is that  
\begin{subequations}\label{eq:o} 
  \begin{align} 
    \frac{\partial I(\rank(X);\rank(Y))}{\partial Q_r(\tilde 
       U)}  & + 
    R(r)g(\tilde U) = \lambda_r  \nonumber \\ & \forall r, \tilde 
    U\in \Gr(r,\ffield^M):   
    Q_r(\tilde U)>0, \label{eq:oa}  \\ 
    \frac{\partial I(\rank(X);\rank(Y))}{\partial Q_r(\tilde 
       U)}  & + 
     R(r)g(\tilde U)   \leq \lambda_r \nonumber \\ & \forall r, \tilde 
    U\in \Gr(r,\ffield^M) :  
     Q_r(\tilde U)=0, \label{eq:ob}\\ 
    \frac{\partial I(\rank(X);\rank(Y))}{\partial R(r)} & + 
    \sum_{\tilde U \in \Gr(r,\ffield^M)} Q_r(\tilde U) 
    g(\tilde U)  = \bar \lambda  \nonumber \\ &  \forall r:  
    R(r)>0,\label{eq:oc} \\ 
     \frac{\partial I(\rank(X);\rank(Y))}{\partial R(r)} & + 
    \sum_{\tilde U \in \Gr(r,\ffield^M)} Q_r(\tilde U) 
    g(\tilde U)  \leq \bar \lambda  \nonumber \\ &  \forall r:  
    R(r)=0,\label{eq:od} 
  \end{align} 
\end{subequations} 
where the partial derivatives are 
\begin{align*} 
  \lefteqn{\frac{\partial I(\rank(X);\rank(Y))}{\partial 
      Q_r(\tilde U)}} \\ & = R(r) 
  \sum_s P_{\rank(Y)|\lspan{X^\tr}}(s|\tilde U) \log_2 
  \frac{P_{\rank(Y)|\rank(X)}(s|r)}{P_{\rank(Y)}(s)} - \log_2 e, 
\end{align*} 
and 
\begin{align*} 
  \lefteqn{\frac{\partial I(\rank(X);\rank(Y))}{\partial 
      R(r)}} \\ & = \sum_s 
  P_{\rank(Y)|\rank(X)}(s|r) \log_2  
  \frac{P_{\rank(Y)|\rank(X)}(s|r)}{P_{\rank(Y)}(s)} - \log_2 e. 
\end{align*} 
We can check that 
\begin{equation*} 
  C_{\subs}(H,T) = \lambda + \log_2 e,  
\end{equation*} 
and 
\begin{equation*} 
  \lambda = \sum_{r}\lambda_r + (M-1)\log_2 e. 
\end{equation*}

The above optimization problem to find an optimal input for 
subspace coding is in general hard.  
We have already shown that for large $T$, the $M$-dimensional 
$\alpha$-type input gives a good approximation of the 
channel capacity (see Prop.~\ref{lemma:low}). 
Here, we can further improve the result for a class of LOCs 
 
\begin{definition} 
  A random matrix $H$  with dimension $M\times N$ is  
  \emph{regular} if $p_{\rank(H)}(s)>0$ for $0\leq s\leq M$. 
   $\loc(H,T)$ is regular if $H$ is regular. 
\end{definition} 
 
\begin{theorem}\label{the:mmm} 
  Consider regular $\loc(H,T)$ with dimension $M\times N$.  
  There exists $T_1$ such that when 
  $T\geq T_1$, $C_{\subs}$ is achieved by the $\alpha$-type 
  input with $R(M) = 1$. 
  In this case $C_{\subs}(H,T) = 
  g(\ffield^M) = \sum_s p_{\rank(H)}(s) 
  \log_2\frac{\cmat{T}{s}}{\cmat{M}{s}} = 
  \E\left[\log_2\frac{\cmat{T}{\rank(H)}}{\cmat{M}{\rank(H)}}\right]$.  
\end{theorem} 
\begin{IEEEproof} 
  See \S\ref{sec:p9}. 
\end{IEEEproof}

Assume $M\leq N$. Since $p_{\rank(H_{\text{pure}})}(r) = 
\cmat{M,N}{r} q^{-MN}$ for $r\leq M$, 
 $H_{\text{pure}}$ is 
regular.

\subsection{Optimal Constant-Rank Inputs} 
 
An input for a LOC with $p_{\rank(X)}(r) = 1$ is called a 
\emph{constant-rank or rank-$r$ input distribution}.  When talking 
about subspace coding, rank-$r$ input is corresponding to 
$r$-dimensional subspace coding.  Our discussion of constant-rank 
inputs for subspace coding is equivalent to the discussion of 
constant-dimensional subspace coding.  
 
Let 
\begin{equation*} 
  C_{\csub}(H,T) = \max_{p_X:\text{constant-rank}}I(\lspan{X};\lspan{Y}) 
\end{equation*} 
so that $C_{\csub}(H,T)$ is the maximum achievable rates of 
constant-dimensional subspace coding. 
Let $\bar C_{\csub}(H,T)$ be the normalization of  
$C_{\csub}(H,T)$ by $T\log_2 q$. 
The rank of a constant-rank input that achieves $C_{\csub}(H,T)$ is called an 
\emph{optimal input rank}. 
We show that to find an optimal input rank, we only need to consider 
$\alpha$-type inputs. Moreover, we can determine 
$C_{\csub}(H,T)$ and an optimal 
input rank based on sufficient channel statistics such that 
we can calculate $g(\tilde U)$. See Prop.~\ref{lemma:dase} 
and Theorem~\ref{the:but} for details.

\begin{theorem}\label{lemma:dase} 
  For any LOCs, there exists a constant-rank $\alpha$-type input 
that achieves $C_{\csub}(H,T)$. 
\end{theorem} 
\begin{IEEEproof} 
  The proof is similar to the proof of 
  Proposition~\ref{the:diq}. See \S\ref{sec:p9}. 
\end{IEEEproof}

\begin{theorem}\label{the:but} 
  For $\loc(H,T)$ with dimension $M\times N$,  
  let 
  \begin{equation*} 
    U^* = \arg \max_{\tilde U\in \Pj(M^*,\ffield^M)} 
    g(\tilde U). 
  \end{equation*} 
  Then, $r^*=\dim(\tilde U^*)$ is an optimal input rank and 
  $C_{\csub}(H,T)= g(\tilde U^*)$. Furthermore, 
  \begin{equation*} 
    \bar C_{\subs}(H,T)- \bar C_{\csub}(H,T) \leq \frac{\log_2 
      \min\{M,N\}}{T\log_2 q}. 
  \end{equation*} 
\end{theorem} 
\begin{IEEEproof} 
  See \S\ref{sec:p9}. 
\end{IEEEproof}

Theorem~\ref{the:but} also  
bounds the loss of rate when using constant-dimensional 
subspace coding.  
Assume $M=N=5$, $T=10$, $q=4$ and $\E[\rank(H)] = M/4$. We have 
\begin{align*} 
  \frac{\bar C_{\subs}(H,T)- \bar C_{\csub}(H,T)}{\bar 
    C_{\subs}(H,T)} &  < \frac{\log_2 M}{T\log_2 q 
    (1-M/T)\E[\rank(H)]} \\ 
  & = 9.8\
\end{align*} 
So the loss of rate is marginal for typical channel 
parameters.

\subsection{Optimal Input Rank} 
 
To evaluate the results in Theorems~\ref{lemma:dase} 
and Theorem~\ref{the:but}, we require the distribution of the 
transformation matrix. Now we show that in some cases, we 
can relax this requirement significantly.  
For $\loc(H,T)$, recall that  
$$\rank^*(H)=\max\{r:\Pr\{\rank(H)=r\}>0\}.$$ 
 
\begin{theorem}\label{the:also} 
  For $\loc(H,T)$, there 
  exists $T_0$ such that when $T\geq T_0$, $r^* 
  \geq \rank^*(H)$, where $r^*$ is the optimal input rank given in 
  Theorem~\ref{the:but}. 
\end{theorem} 
\begin{IEEEproof} 
  Suppose the dimension of the LOC is $M\times N$. 
  Fix $T_0$ such that $\Theta(T_0,r,H)\geq 0$ for all $r< 
  \rank^*(H)$. This is possible because $\Theta(T,r,H)$ is a 
  linearly increase function of $T$ for all $r< 
  \rank^*(H)$. 
  Assume $T\geq T_0$ and $ r^* < \rank^*(H)$. 
  For any $\tilde V\leq 
  \ffield^M$ with $\dim(\tilde V)<\rank^*(H)\leq M$, by 
  Lemma~\ref{lemma:food2}, 
  \begin{equation*} 
    g(\ffield^M) >  g(\tilde V). 
  \end{equation*} 
  Thus, we have a contradiction to $r^* 
  < \rank^*(H)$. 
\end{IEEEproof} 
 
Theorem~\ref{the:also} narrows down the range to search 
an optimal input rank for large $T$. When $\rank^*(H)=M$, it 
tells that $M$ is an optimal input rank when $T$ is 
sufficiently large. The proof of Theorem~\ref{the:also} 
tells how to find a $T_0$.  
 
As an example, let us check the optimal input rank of 
$\loc(H_{\text{full}},T)$. 
We know $\rank^*(H_{\text{full}}) = M$ and 
$p_{\rank(H_{\text{full}})}(M)=1$. 
By Theorem~\ref{the:also}, there 
exists $T_0$ such that when $T>T_0$, $r^*=M$.  
Now we want to know the value of $T_0$. 
From the proof of 
Theorem~\ref{the:also}, we know that $T_0$ should satisfy 
\begin{equation*} 
  \Theta(T_0,r,H_{\text{full}})\geq 0,\quad \forall r<M. 
\end{equation*} 
In other words, $T_0$ should satisfy 
\begin{equation} 
  \label{eq:alt} 
  (r-T_0/2)^2-(T_0/2-M)^2 + \log_q \cmatt{r}{r}\geq 0,\quad \forall r<M. 
\end{equation} 
If $M\leq T_0\leq 2M-1$, (\ref{eq:alt}) does not hold for 
$r=M-1$. If $T_0= 2M$, the minimum value of the RHS of 
(\ref{eq:alt}) is obtained for $r=M-1$, i.e., $1 + \log_q 
\cmatt{M}{M}$, which is positive when $q>2$.  Similarly, we 
can check that $T_0=2M+1$ is sufficient for any field size. 
As a conclusion, when i) $q>2$ and $T\geq 2M$ or ii) $q=2$ 
and $T\geq 2M+1$, the $M$-dimensional $\alpha$-type input is an 
optimal constant-rank input for $(H_{\text{full}},T)$.

\subsection{proofs} 
\label{sec:p9}

\begin{IEEEproof}[Proof of Lemma \ref{lemma:food2}] 
  Let $\tilde U = \ffield^M$. 
  Since $\tilde V\leq \tilde U$, we can find a full rank 
  $M\times M$ matrix 
  \begin{equation*} 
    \mathbf D = \begin{bmatrix} \mathbf D_0 \\ \mathbf  D_1 \end{bmatrix} 
  \end{equation*} 
  such that $\lspan{\mathbf{D}^\tr} = \tilde U$ and 
  $\lspan{\mathbf{D}_1^\tr} = \tilde V$. By 
  Coro.~\ref{lemma:cond},  
  \begin{align*} 
    \sum_{s'\geq s} P_{\rank(Y)|\lspan{X^\tr}}(s'|\tilde V)  & =  
    \Pr\{\rank(\mathbf{D_1} H)\geq s\}, 
  \end{align*} 
  and  
  \begin{align*} 
    P_{\rank(Y)|\lspan{X^\tr}}(s|\tilde U) & = 
    \Pr\{\rank(\mathbf D H)= s\} \\  & = \Pr\{\rank(H) = s\}. 
  \end{align*} 
  We know $\Pr\{\rank(H) \geq s\} \geq 
  \Pr\{\rank(\mathbf{D_1} H)\geq s\}$. So 
  \begin{equation*} 
    \sum_{s'\geq s} P_{\rank(Y)|\lspan{X^\tr}}(s'|\tilde U) 
    \geq \sum_{s'\geq s} P_{\rank(Y)|\lspan{X^\tr}}(s'|\tilde V). 
  \end{equation*} 
  Moreover, for $s$ such that $r <s\leq \rank^*(H)$, 
  \begin{equation*} 
    \sum_{s'\geq s} P_{\rank(Y)|\lspan{X^\tr}}(s'|\tilde V) 
    = 0. 
  \end{equation*} 
  Thus,  
  \begin{align} 
    \lefteqn{\sum_{s}s (P_{\rank(Y)|\lspan{X^\tr}}(s|\tilde U) - 
    P_{\rank(Y)|\lspan{X^\tr}}(s|\tilde V))} \nonumber \\ & = \sum_k 
    \sum_{s\geq k} (P_{\rank(Y)|\lspan{X^\tr}}(s|\tilde U) - 
    P_{\rank(Y)|\lspan{X^\tr}}(s|\tilde V))\nonumber \\ & \geq 
    \sum_{k: \rank^*(H)\geq k>r}\sum_{s: s\geq k} 
    P_{\rank(Y)|\lspan{X^\tr}}(s|\tilde U) \nonumber \\ & \geq  
    \sum_{k: \rank^*(H)\geq k>r} \Pr\{\rank(H)=\rank^*(H)\} \nonumber \\ & = 
    (\rank^*(H)-r) \Pr\{\rank(H)=\rank^*(H)\}. \label{eq:in1} 
  \end{align}

  By definition, 
  \begin{align*} 
    \lefteqn{\frac{g(\tilde U) - g(\tilde V)}{\log_2 q}} \nonumber  
    \\ & = \sum_{s} 
    P_{\rank(Y)|\lspan{X^\tr}}(s|\tilde U) \left((T-M)s+\log_q 
      \frac{\cmatt{T}{s}}{\cmatt{M}{s}}\right) \nonumber \\ 
    & \qquad - \sum_s P_{\rank(Y)|\lspan{X^\tr}}(s|\tilde 
    V) \left((T-r)s+\log_q 
      \frac{\cmatt{T}{s}}{\cmatt{r}{s}}\right)  \nonumber \\ 
    & = (T-M) \sum_{s}s (P_{\rank(Y)|\lspan{X^\tr}}(s|\tilde 
    U) -P_{\rank(Y)|\lspan{X^\tr}}(s|\tilde V)) \nonumber \\ & \qquad 
    - (M-r)\sum_{s}sP_{\rank(Y)|\lspan{X^\tr}}(s|\tilde V) \nonumber  \\ & 
    \qquad + 
    \sum_s P_{\rank(Y)|\lspan{X^\tr}}(s|\tilde U) \log_q 
    \frac{\cmatt{T}{s}}{\cmatt{M}{s}} \nonumber \\ & \qquad - \sum_s 
    P_{\rank(Y)|\lspan{X^\tr}}(s|\tilde V) \log_q 
    \frac{\cmatt{T}{s}}{\cmatt{r}{s}} \nonumber \\ 
    & > (T-M)(\rank^*(H)-r) \Pr\{\rank(H)=\rank^*(H)\} \nonumber\\ & 
    \qquad - r(M-r) + \log_q \cmatt{r}{r}, 
  \end{align*} 
  where the last inequality follows from \eqref{eq:in1}. Therefore 
  \begin{equation*} 
    (M-r)\sum_{s}sP_{\rank(Y)|\lspan{X^\tr}}(s|\tilde V) \leq r(M-r), 
  \end{equation*} 
  \begin{equation*} 
     \sum_s P_{\rank(Y)|\lspan{X^\tr}}(s|\tilde U) \log_q 
    \frac{\cmatt{T}{s}}{\cmatt{M}{s}} \geq 0, 
  \end{equation*} 
  and  
  \begin{equation*} 
    \sum_s 
    P_{\rank(Y)|\lspan{X^\tr}}(s|\tilde V) \log_q 
    \frac{\cmatt{T}{s}}{\cmatt{r}{s}} < \sum_s 
    P_{\rank(Y)|\lspan{X^\tr}}(s|\tilde V) \log_q 
    \frac{1}{\cmatt{r}{s}} \leq \log_q 
    \frac{1}{\cmatt{r}{r}}. 
  \end{equation*} 
\end{IEEEproof}

\begin{IEEEproof}[Proof of Theorem \ref{the:mmm}] 
To prove the theorem, we only need to check that the 
$\alpha$-type input with $R(M) = 1$ satisfies 
(\ref{eq:o}). 
Conditions (\ref{eq:oa}) and (\ref{eq:ob}) with  
$r<M$ are satisfied by $\lambda_r=\log_2 e$ because $R(r)=0$. 
Since $Q_M(\ffield^M)=1$, we check  
condition (\ref{eq:oa}) with $r=M$. Since 
$P_{\rank(Y)|\rank(X)}(s|M) = P_{\rank(Y)}(s)$, 
\begin{equation*} 
  \left. \frac{\partial I(\rank(X);\rank(Y))}{\partial 
      Q_M(\ffield^M)}\right |_{R(M)  = 1} = -\log_2 e. 
\end{equation*} 
So, (\ref{eq:oa}) with $r=M$ is satisfied by $\lambda_M=g(\ffield^M)-\log_2 e$. 
This completes the verification of (\ref{eq:oa}) and (\ref{eq:ob}).

The above analysis also tells that $\bar \lambda = \lambda_M$.  
Now we check (\ref{eq:oc}) and (\ref{eq:od}) with $\bar \lambda = 
g(\ffield^M) - \log_2 e$.  
Since $R(M)=1$, condition (\ref{eq:oc}) should be satisfied with $r=M$. 
This is true since  
\begin{equation*} 
  \left. \frac{\partial I(\rank(X);\rank(Y))}{\partial 
      R(M)}\right|_{R(M)  = 1} + g(\ffield^M) = -\log_2 e + g(\ffield^M). 
\end{equation*} 
Next, we check condition (\ref{eq:od}) for $r<M$. We know 
\begin{align*} 
  \lefteqn{\left. \frac{\partial I(\rank(X);\rank(Y))}{\partial 
      R(r)}\right|_{R(M)  = 1}} \\  & = \underbrace{\sum_s 
  P_{\rank(Y)|\rank(X)}(s|r) \log_2  
  \frac{P_{\rank(Y)|\rank(X)}(s|r)}{P_{\rank(Y)|\rank(X)}(s|M)}}_{(A)} - 
  \log_2 e. 
\end{align*} 
  Since 
  \begin{align*} 
    P_{\rank(Y)|\rank(X)}(s|M) & = 
    P_{\rank(Y)|\lspan{X^\tr}}(s|\ffield^M) \\ 
    & = \Pr\{\rank(\mathbf D H)= s\} \\ 
    & = \Pr\{\rank(H) = s\}, 
  \end{align*} 
  we have 
  \begin{align*} 
    (A)  
    & \leq \sum_{s}  P_{\rank(Y)|\rank(X)}(s|r) \log_2  
  \frac{1}{P_{\rank(Y)|\rank(X)}(s|M)}  \\ 
    & = \sum_{s} P_{\rank(Y)|\rank(X)}(s|r) \log_2  
  \frac{1}{p_{\rank(H)}(s)} \\  
    & \leq - \log_2 \min_{0\leq s< M}p_{\rank(H)}(s).  
  \end{align*} 
  That is 
  \begin{equation*} 
    \left. \frac{\partial I(\rank(X);\rank(Y))}{\partial 
      R(r)}\right|_{R(M)  = 1} \leq - \log_2 \min_{0\leq s\leq 
    M}p_{\rank(H)}(s) -\log_2 e. 
  \end{equation*} 
  Fix $T_1$ such that 
  $\Theta(T_1,r,H) \geq - \log_2 \min_{0\leq 
      s< M}p_{\rank(H)}(s)$ for all $r< M$.  
  This is possible because $\Theta(T,r,H)$ is linearly increase with 
  $T$ and $ - \log_2 \min_{0\leq s< M}p_{\rank(H)}(s)$ does not change 
  with $T$. 
  By Lemma~\ref{lemma:food2},  
  $g(\ffield^M)\geq g(\tilde U) - \log_2 \min_{0\leq 
      s\leq M}p_{\rank(H)}(s)$ for all $\tilde U\in 
    \Gr(r,\ffield^M)$. Thus 
    \begin{align*} 
      \bar \lambda & = g(\ffield^M) - \log_2 e \\ 
      & \geq \max_{\tilde U\in 
    \Gr(r,\ffield^M)} g(\tilde U) - \log_2 \min_{0\leq 
      s\leq M}p_{\rank(H)}(s) - \log_2 e \\ 
      & \geq \sum_{\tilde U\in 
    \Gr(r,\ffield^M)} Q_r(\tilde U) g(\tilde U) + \left. \frac{\partial I(\rank(X);\rank(Y))}{\partial 
      R(r)}\right|_{R(M)  = 1}.  
    \end{align*} 
Hence, condition (\ref{eq:od}) with $r<M$ is satisfied. 
\end{IEEEproof}

\begin{IEEEproof}[Proof of Theorem~\ref{lemma:dase}] 
  Consider a LOC with block length $T$.  Let $p_X(\bX)$ be an optimal 
  constant-rank input with $p_{\rank(X)}(r^*)=1$. For $\Phi\in 
  \Fr(\ffield^{T\times T})$, define $p^{\Phi}$ as $p_X^{\Phi}(\bX) = 
  p_X(\Phi\bX)$. It is clear that $p^{\Phi}_{\rank(X}(r^*)=1$.  By 
  Lemma~\ref{lemma:mc1}, $p_X^{\Phi}(\bX)$ is also an optimal 
  constant-rank input.  Define $p_X^*$ as $$p_X^*(\bX) = 
  \frac{1}{|\Fr(\ffield^{T\times T})|}\sum_{\Phi\in 
    \Fr(\ffield^{T\times T})} p_X^{\Phi}(\bX).$$ By the concavity of 
  the mutual information, we know $p_X^*$ is also an optimal 
  constant-rank input.  We can check that $p_X^*$ is 
  $\alpha$-type as in the proof of Proposition~\ref{the:diq}. 
\end{IEEEproof} 
 
\begin{IEEEproof}[Proof of Theorem~\ref{the:but}] 
  For an $r$-dimensional $\alpha$-type input, 
  \begin{align*} 
    I(\lspan{X};\lspan{Y}) & = \sum_{\tilde U\in 
      \Gr(r,\ffield^M)} Q_r(\tilde U) g(\tilde U) \\ 
    & \leq \max_{\tilde U\in \Gr(r,\ffield^M)} g(\tilde U) 
    \\ & \leq g(\tilde U^*). 
  \end{align*} 
  Thus $C_{\csub} \leq g(\tilde U^*)$. On the other hand, for the 
  $r^*$-dimensional $\alpha$-type input with $p_{\lspan{X^\tr}}(\tilde 
  U^*) = 1$, $C_{\csub} \geq I(\lspan{X};\lspan{Y}) = g(\tilde U^*)$. 
 
  Furthermore, for an $\alpha$-type input  
  \begin{align*} 
    I(\lspan{X};\lspan{Y}) - C_{\csub} & =  
      I(\rank(X);\rank(Y))+ J(\rank{X};\rank{Y})- 
     g(\tilde U^*) \\ & \leq I(\rank(X);\rank(Y)) \\ & \leq \log_2 \min\{M,N\}. 
  \end{align*} 
  Thus, $C_{\subs}- C_{\csub} = \max_{p_X:\alpha-\text{type}} I(\lspan{X};\lspan{Y}) - C_{\csub} \leq \log_2 \min\{M,N\}$. 
\end{IEEEproof}

\section{Coding for Linear Operator Channels} 
\label{sec:codes} 
 
From this section, we consider coding design for $\loc(H,T)$. 
 
\subsection{Using Channel Training or Subspace Coding} 
 
We have considered two kinds of coding schemes for noncoherent LOCs: 
channel training and subspace coding. 
For channel training, all the input matrices $\bX$ have the form 
\begin{equation}\label{eq:CT} 
	\bX = \begin{bmatrix} \mathbf{I} \\ \tilde \bX \end{bmatrix}, 
\end{equation} 
where $\mathbf I$ is an identity matrix.  
For such a transmission, the received matrix  
\begin{equation*} 
	\bY = \begin{bmatrix} \mathbf{I} \\ \tilde \bX  \end{bmatrix} \bH = \begin{bmatrix} \bH \\ \tilde\bX\bH \end{bmatrix}, 
\end{equation*} 
where $\bH$ is the instance of $H$.  The receiver can use the first 
part of $\bY$ to recover $\bH$ and use this information to decode 
$\tilde \bX$.  We have shown that the normalized maximum achievable 
rate using channel training is 
\begin{equation*} 
  \bar C_{\text{CT}}(H,T) = (1-M/T)\E[\rank(H)]. 
\end{equation*}

For subspace coding, a codeword contains a sequence of subspaces and 
the transmission of a subspace through LOCs involves the transformation 
of a subspace to a matrix.  The decoding also only depends on the 
subspace spanned by the received matrix.  For more details, see our discuss 
in \S\ref{sec:subspa}. 
We have shown that the normalized maximum achievable 
rate using subspace coding satisfies 
\begin{equation*} 
  \E[\rank(H)] \geq \bar C_{\subs}(H,T) \geq (1-M/T)\E[\rank(H)] + \epsilon(T,q), 
\end{equation*} 
where $0< \epsilon(T,q) < 1.8/(T\log_2q)$.  We have shown the lower bound 
of $\bar C_{\subs}(H,T)$ is tight for regular LOCs when $T$ is 
sufficiently large. 
Therefore 
\begin{equation*} 
  \epsilon(T,q) \leq   \bar C_{\subs}(H,T) - \bar C_{\text{CT}}(H,T) < M/T\E[\rank(H)].  
\end{equation*} 
So using channel training does not loss much in rates, especially when 
$T$ is large. 
 
For encoding, a channel training code can be regarded 
as a special subspace code. But the decoding of channel training codes 
uses the received matrices, while the decoding of subspace codes uses 
the subspaces spanned by the received matrices. However, we can just 
decode a subspace code using the matrices received. If we apply this 
decoding method of subspace codes, channel training can be regarded as 
a special case of subspace coding.  This is the reason that even some 
existing subspace coding schemes use channel training \cite{silva08j, 
  nobrega10}. 
 
In this paper, we only study the design of channel training codes.

\subsection{Some Existing Channel Training Codes} 
 
Existing coding schemes for RLCN also works for LOCs, even though a 
RLCN is a special LOC with its transformation matrix depends on the 
network topology.  In fact, most coding practice of RLCN is based on 
channel training. We first introduce two coding schemes for RLCN.

The first coding scheme was introduced by Ho \etal 
\cite{ho06j}. They assumed that the transformation matrix has rank $M$. 
In their scheme, a codeword has the form in \eqref{eq:CT} 
where any matrix in $\ffield^{(T-M)\times M}$ can be used as $\tilde \bX$. 
We call such codes the \emph{classical channel training codes}. 
A received matrix has the form 
\begin{equation*} 
  \bY = \begin{bmatrix} \mathbf{I} \\ \tilde\bX \end{bmatrix} \bH = \begin{bmatrix} \bH \\ \tilde\bX\bH \end{bmatrix}.  
\end{equation*} 
Since $\bH$ has rank $M$, the receiver can decode $\tilde \bX$ by solving a 
system of linear equations.  
The rateless realization of random linear network codes found in 
\cite{chachu07,Katti08} applies a classical channel training code over $\loc(GH,T)$, where $\loc(H,T)$ is the original channel and $G$ is an $r\times M$ purely random matrix.  
We will give a general discussion of this approach and show that we only need to consider $r<M$.

Silva \etal\ \cite{silva08j} proposed a more general method in which 
$\tilde \bX$ in \eqref{eq:CT} can only be chosen from a rank-metric code.  
The redundancy in the rank-metric code can be used 
to correct the rank deficiency of $H$ as well as additive errors, 
which are not considered in this work.  This code construction is nearly optimal in 
terms of a Singleton type coding bound on one-block subspace codes  
\cite{koetter08j}.

Both of the works \cite{ho06j, silva08j} construct channel training codes with unit 
block, which in general cannot achieve the channel capacity of LOCs.  
Two more recent works \cite{nobrega09, nobrega10} considers design of 
channel training codes with non-unit length. The authors proposed a 
multilevel code construction approach in \cite{nobrega09}.  
Parallel to our work, this approach is used explicitly 
to construct ``multishot rank-metric codes'' \cite{nobrega09}.  
Note that the multishot 
rank-metric codes constructed in \cite{nobrega09} is different to the 
codes we will proposed here, even though we both apply rank 
metric. For 
the lack of a performance evaluation of their codes, we cannot see whether 
their codes achieve $\bar C_{\text{CT}}$.

\subsection{Achieve Higher Rate than $\bar C_{\text{CT}}$} 
\label{sec:high} 
 
In the following, we will introduce two constructions of channel 
training codes for LOCs, called lifted rank-metric codes and lifted 
linear matrix codes, respectively.  We will prove that lifted linear 
matrix codes can achieve $\bar C_{\text{CT}}$. 
But our codes can also used to achieve higher 
rate than $\bar C_{\text{CT}}$ using extended channel training.  
The approach is to use $\loc(H,T)$ 
as $\loc(GH,T)$ for any $r\times M$ random matrix $G$.  
As we discussed in \S\ref{sec:ct}, we only need to consider $r<M$ and 
Theorem~\ref{the:ect} implies that a purely random $G$ is good enough.

\subsection{Formulation of Channel Training Codes} 
 
A matrix code $\lc^{(n)}\subset \ffield^{(T-M)\times nM}$ induces a 
channel training code for $\loc(H,T)$ with dimension $M\times N$ as 
follows.  For $\tilde \bX^{(n)}\in \lc^{(n)}$, we write 
\begin{align*}
  \tilde\bX^{(n)} = \begin{bmatrix} \tilde\bX_1& \tilde\bX_2&\cdots& \tilde\bX_n\end{bmatrix}, 
\end{align*} 
where $\bX_i\in \ffield^{(T-M)\times M}$. 
Define the $M$-lifting of $\tilde\bX^{(n)}$, which extends the definition 
of lifting in \cite{silva08j}, as  
\begin{align*} 
  L_M( \tilde\bX^{(n)}) = \left(\begin{bmatrix}\mathbf I_M 
      \\  \tilde\bX_1 \end{bmatrix},  \begin{bmatrix}\mathbf I_M 
      \\  \tilde\bX_2 \end{bmatrix},  
      \cdots, \begin{bmatrix}\mathbf I_M 
      \\  \tilde\bX_n \end{bmatrix} \right), 
\end{align*} 
where $\mathbf I_M$ is an $M\times M$ identity matrix.  We see 
$L_M(\tilde\bX^{(n)})\in (\ffield^{T\times M})^n$.  Define the $M$-lifting 
of $\lc^{(n)}$ as 
\begin{align*}
  L_M(\lc^{(n)}) = \{L_M( \tilde\bX^{(n)}): \tilde \bX^{(n)}\in \lc^{(n)}\}. 
\end{align*} 
We call $L_M(\lc^{(n)})$ the \emph{lifted matrix code} of 
$\lc^{(n)}$. 
When the context is clear, we write $L(\tilde \bX^{(n)})$ for 
$L_M(\tilde \bX^{(n)})$ and $L(\lc^{(n)})$ for $L_M(\lc^{(n)})$. 
The rate 
$\mathcal{R}^{(n)}$ of $L(\lc^{(n)})$ is 
\begin{align*} 
  \mathcal{R}^{(n)}  = 
  \frac{\log_2|L(\lc^{(n)})|}{nT\log_2 q} = 
  \frac{\log_2|\lc^{(n)}|}{nT\log_2 q}. 
\end{align*}

Suppose that the transmitted codeword is $L(\tilde\bX^{(n)})$. 
Each use of $\loc(H,T)$ can 
transmit one component of $L(\tilde\bX^{(n)})$.  The $i$th output 
matrix of $\loc(H,T)$ is 
\begin{align}\label{eq:output} 
   \bY_i = \begin{bmatrix}\mathbf I_M \\\tilde\bX_{i}\end{bmatrix} 
  \bH_i
   = \begin{bmatrix} \bH_i \\ \tilde\bY_i \end{bmatrix}, 
\end{align} 
where $\bH_i$ is the $i$th instance of $H$ and $\tilde\bY_i=\tilde\bX_{i}\bH_i$.  Let 
\begin{equation*}
  \bH^{(n)} = \begin{bmatrix} \bH_1 & & & & \\ & \bH_2 & & & \\ & & \ddots 
    & & \\ & & & \bH_n \end{bmatrix},  
\end{equation*} 
and  
\begin{equation*}
  \tilde\bY^{(n)} = \begin{bmatrix} \tilde\bY_1& \tilde\bY_2&\cdots& \tilde\bY_n\end{bmatrix}. 
\end{equation*} 
We obtain the decoding equation of the lifted matrix code $L(\lc^{(n)})$ as  
\begin{equation} 
  \label{eq:dec} 
  \tilde\bY^{(n)} =  \tilde\bX^{(n)}\bH^{(n)}. 
\end{equation} 
The decoding of $\tilde\bY^{(n)}$ can use the knowledge of $\bH^{(n)}$.

\section{Rank-Metric Codes for LOCs} 
\label{sec:rank} 
 
In this section, we extend the rank-metric approach of Silva \etal\ 
to construct matrix codes for LOCs. 
 
\subsection{Rank-Metric Codes} 
 
Define the \emph{rank distance} between $\bX,\bX'\in \ffield^{t\times m}$ as 
\begin{equation*} 
  d(\bX,\bX')= \rank(\bX-\bX'). 
\end{equation*} 
A rank metric code is a unit-length matrix code with the 
rank distance \cite{gabidulin85}. 
The minimum distance of a rank-metric code 
$\mathcal{C}\subset \ffield^{t\times m}$ is 
\begin{equation*} 
  D(\mathcal{C}) = \min_{\bX\neq \bX'\in\mathcal{C}} d(\bX,\bX'). 
\end{equation*} 
When $t\geq m$, we have 
\begin{equation}\label{eq:singleton} 
   \frac{\log_2 |\mathcal{C}|}{t \log_2 q} \leq m-D(\mathcal{C})+1, 
\end{equation} 
which is called the Singleton bound for rank-metric codes 
\cite{gabidulin85} (see also \cite{silva08j} and the reference 
therein).  Codes that achieve this bound are called 
\emph{maximum-rank-distance (MRD)} codes. Gabidulin described a class 
of MRD codes for $t\geq m$, which are analogs of generalized 
Reed-Solomon codes \cite{gabidulin85}.

Suppose the transmitted codeword is $\bX_0\in \lc$ and the received 
matrix is $\bY=\bX_0\bH$. If $\bH$ is known at the receiver, we can 
decode $\bY$ using the minimum distance decoder defined as 
\begin{equation}\label{eq:mdd} 
  \hat{\bX} = \arg\min_{\bX \in  \lc} d(\bY, \bX\bH).  
\end{equation}  
 
\begin{proposition}\label{lemma:iif} 
  The minimum distance decoder is 
  guaranteed to return $\hat{\bX}=\bX_0$ for all 
  $\bH$ with $\rank(\bH) \geq r$ if and only if 
  $D(\lc) \geq m-r+1$, where $0<r\leq m$. 
\end{proposition} 
\textbf{Remark:} Silva \etal\ only proved the sufficient condition in 
Prop.~\ref{lemma:iif} when considering additive errors. In fact, the necessary condition also holds without considering the additive errors 
as \cite{koetter08j, silva08j}.   
\begin{IEEEproof} 
 We first prove the sufficient condition. Assume 
 $D(\lc) \geq m-r+1$ and $\rank(\bH)\geq r$. We know 
 $d(\bY,\bX_0\bH)=0$. Suppose that there exists a 
 different codeword $\bX_1\in\mathcal{C}$ with 
 $d(\bY,\bX_1\bH)=0$. We have $(\bX_0-\bX_1) 
 \bH = \bzero$.  Using the rank-nullity theorem of linear algebra, 
 $d(\bX_0,\bX_1) = \rank(\bX_0-\bX_1) \leq M- 
 \rank(\bH) \leq m - r$, i.e., a contradiction to $D(\lc) \geq m -r+1$.

Now we prove the necessary condition. Assume 
$D(\lc) \leq m-r$. There must exist $\bX_1, 
\bX_2\in\mathcal{C}$ such that 
$d(\bX_1,\bX_2)=\rank(\bX_1-\bX_2)\leq m-r$.   
Let 
\begin{equation*} 
 B = \{\mathbf h \in \ffield^{m\times 1}: (\bX_1-\bX_2) 
 \mathbf h = \bzero\}.  
\end{equation*} 
We know $\dim(B)= m-\rank(\bX_1-\bX_2) \geq r$. 
By juxtaposing the vectors in $B$, we can obtain a 
matrix $\bH$ with $\rank(\bH)\geq r$. 
We know $(\bX_1-\bX_2)\bH=\bzero$. So if the 
transformation matrix is $\bH$, the decoder cannot always 
output the correct codeword. 
\end{IEEEproof}

\subsection{Lifted Rank-Metric Codes} 
 
Consider $\loc(H,T)$ with dimension $M\times N$.  The lifted matrix 
codes $L(\lc^{(n)})$, where $\lc^{(n)}\in\ffield^{(T-M)\times n M}$ is 
a rank-metric code, is also called \emph{lifted rank-metric code}. 
The unit-block (one-shot) lifted rank-metric code ($n=1$) is first 
used by Silva \etal\ in random network coding \cite{silva08j}. Here we 
extend their approach to multiple usages of the channel. 
 
By the Singleton bound of rank-metric codes in 
\eqref{eq:singleton}, 
\begin{equation*} 
  \frac{\log_2 |\lc^{(n)}|}{(T-M) \log_2 
  q} \leq nM- D(\lc^{(n)})+1. 
\end{equation*} 
Thus the rate of $L_M(\lc^{(n)})$ 
\begin{align} 
  \mathcal{R}^{(n)} & = \frac{\log_2 |\lc^{(n)}|}{nT\log_2q} \nonumber \\ & \locequa{\leq} \frac{(nM- D(\lc^{(n)})+1)(T-M)\log_2 q}{nT\log_2 q} \nonumber \\ 
   & = (1-M/T) (M- D(\lc^{(n)})/n+1/n), \label{eq:free} 
\end{align} 
where the equality in (a) is achieved by MRD codes.

Suppose that the transmitted codeword is $L(\tilde\bX_0^{(n)})$.  
By the decoding equality in \eqref{eq:dec},  
we can decode $\tilde\bY^{(n)}$ 
using the minimum distance decoder defined in \eqref{eq:mdd}. 
By Prop. \ref{lemma:iif}, the minimum distance decoder is 
  guaranteed to return $\hat{\bX}^{(n)}=\tilde\bX_0^{(n)}$ for all 
  $\bH^{(n)}$ with $\rank(\bH^{(n)}) \geq nM- D(\lc^{(n)})+1$.

\subsection{Throughput of Lifted Rank-Metric Codes} 
 
Let
\begin{equation}\label{eq:rh} 
  H^{(n)}= \begin{bmatrix} H_1 & & & & \\ & H_2 & & & \\ & & \ddots 
    & & \\ & & & H_n \end{bmatrix}, 
\end{equation} 
in which $H_i$, $i=1,\cdots, n$, are independent and follow the same 
distribution of $H$.  By our analysis above, a receiver using the minimum 
distance decoder can judge if its decoding is guaranteed to be correct 
by checking $\rank(\bH^{(n)})$, which is an instance of $H^{(n)}$.  If 
$\rank(\bH^{(n)})\geq nM- D(\lc^{(n)})+1$, the decoding is guaranteed 
to be correction. Otherwise, if $\rank(\bH^{(n)})< nM- D(\lc^{(n)})+1$, correct decoding cannot be guaranteed. 
Define the \emph{throughput} of $L(\lc^{(n)})$ as 
\begin{align*} 
  TP_{\text{RM}}(\lc^{(n)}) \triangleq  \mathcal{R}^{(n)} 
\Pr\{\rank(H^{(n)})\geq nM- D(\lc^{(n)})+1\}, 
\end{align*} 
where $\text{RM}$ stands for rank metric.  Note that this is the 
zero-error maximum achievable rate of lifted rank-metric codes. For 
any rate higher than $TP_{\text{RM}}(\lc^{(n)})$, we cannot guarantee 
error-free decoding.

Since lifted rank-metric 
codes are special channel training coding method, we have 
$TP_{\text{RM}}(\lc^{(n)})\leq \bar C_{\text{CT}}(H,T)$.  
Now we 
look at whether lifted rank-metric codes achieve $\bar C_{\text{CT}}(H,T)$. 
 
\begin{theorem}\label{the:a} 
  For any positive integer $n$, 
\begin{equation}\label{eq:iaa} 
  \max_{\lc^{(n)}\subset \ffield^{(T-M)\times nM}}TP_{\text{RM}}(\lc^{(n)}) \leq \rho^{(n)} \bar C_{\text{CT}}(H,T),  
\end{equation} 
where $\rho^{(n)} \leq 1$ and the equality in \eqref{eq:iaa} holds 
if there exist MRD codes 
$\lc^{(n)} \subset \ffield^{(T-M)\times nM}$ with $D(\lc^{(n)})=nM-r+1$ for $r=1,2,\cdots,nN^*$. 
Moreover, i) $\rho^{(n)} = 1$ if 
and only if $H$ has a constant rank;  ii) $\lim_{n\rightarrow 
  \infty}\rho^{(n)} = 1$. 
\end{theorem} 
\begin{IEEEproof} 
Let $N^*=\min\{M,N\}$, the maximum possible rank of 
$H$. Let $\theta(\lc^{(n)}) = nM-D(\lc^{(n)})+1$.  
By \eqref{eq:free}, 
\begin{align*} 
  TP_{\text{RM}}(\lc^{(n)}) \leq  \left(1-\frac{M}{T}\right)\frac{ \theta(\lc^{(n)})}{n} \Pr\left\{{\rank(H^{(n)})}\geq \theta(\lc^{(n)})\right\},  
\end{align*} 
where the equality holds for MRD codes. Thus 
\begin{align*} 
  \frac{\max_{\lc^{(n)}\subset \ffield^{(T-M)\times M}}TP_{\text{RM}}(\lc^{(n)})}{\bar C_{\text{CT}}(H,T)} & =  \frac{\max_{r\leq nN^*} \max_{\lc^{(n)}\subset 
    \ffield^{(T-M)\times nM}: \theta(\lc^{(n)}) =r}T_{\text{MDD}}(\lc^{(n)})}{(1-M/T) \E[\rank(H)]} \nonumber \\ 
  & \locequa{\leq} \frac{\max_{r\leq nN^*} r \Pr\{\rank(H^{(n)})\geq r\}}{n\E[\rank(H)]} \\ & \triangleq \rho^{(n)},\nonumber 
\end{align*}  
where the equality in (b) holds if there exist MRD codes 
$\lc^{(n)} \subset \ffield^{(T-M)\times nM}$ with $D(\lc^{(n)})=nM-r+1$ for $r=1,2,\cdots,nN^*$.

Now we look at the property of $\rho^{(n)}$.  
For any $0\leq r \leq nN^*$, we have 
\begin{align*} 
  \E[\rank(H^{(n)})]  & = \sum_{s}sp_{\rank(H^{(n)})}(s)\nonumber \\ 
  & \locequa{\geq}  \sum_{s\geq r}sp_{\rank(H^{(n)})}(s) \\ 
  & \locequa{\geq} \sum_{s\geq r}rp_{\rank(H^{(n)})}(s) \\ 
  & = r \Pr\{\rank(H^{(n)})\geq r\}.\nonumber 
\end{align*} 
Thus, $\rho^{(n)}\leq 1$.  
Now we check the condition that $\rho^{(n)}=1$. First, if 
$p_{\rank(H)}(r_0)=1$ for some $0\leq r_0\leq M$, then 
$\rho^{(n)}=1$. Second, if $\E[\rank(H^{(n)})] = r_n 
\Pr\{\rank(H^{(n)})\geq r_n\}$ for some $0\leq r_n \leq nN^*$, then the 
equalities in (c) and (d) hold, which give 
$\Pr\{\rank(H^n) = r_n\} = 1$. Hence, $\Pr\{\rank(H)=r_n/ n\} =1$.

Let $\mu=\E[\rank(H)]$. By the weak law of large numbers, for any 
$\delta>0$ and $\epsilon>0$ there exists $n_0$ such that when $n>n_0$ 
  \begin{align*} 
  	\Pr\{|\rank(H^{(n)})/n-\mu|\leq \delta/2\}\geq 1-\epsilon. 
  \end{align*} 
  Hence, 
  \begin{align*} 
  	\Pr\{\rank(H^{(n)})/n \geq \mu - \delta/2\}\geq 1-\epsilon. 
  \end{align*} 
  Further, for the same $\delta$ when $n>2/\delta$, there exists 
  integer $r_0$ between $n(\mu-\delta)$ and $n(\mu-\delta/2)$. So, 
  when $n>\max\{n_0, 2/\delta\}$, 
  \begin{align*} 
    \rho^{(n)} & \geq \frac{r_0 \Pr\{\rank(H^{(n)}) \geq r_0\}}{n\mu} \\ 
		& \geq \frac{n(\mu-\delta) \Pr\{\rank(H^{(n)})\geq n(\mu-\delta/2)\}}{n\mu} \\ 
		& \geq \frac{(\mu-\delta)(1-\epsilon)}{\mu} \\ 
		& > 1-(\delta/\mu+\epsilon). 
  \end{align*}  
  Therefore, $\lim_{n\rightarrow \infty}\rho^{(n)} = 1$. 
\end{IEEEproof} 
 
We know that when $T-M\geq nM$, for any $0<r\leq nN^*$ MRD code 
$\lc^{(n)}$ with $D(\lc^{(n)}) = nM-r+1$ can be constructed using 
Gabidulin codes\cite{gabidulin85}.  If we use Gabidulin codes the 
equality in \eqref{eq:iaa} holds when $n\leq T/M-1$.  Let us see two 
cases: i) $H$ has a constant rank. Now $\rho^{(1)}=1$. Thus when 
$T\geq 2M$, lifted Gabidulin codes can achieve $\bar 
C_{\text{CT}}$. ii) $H$ has a random rank we require a sufficiently 
large $n_0$ to guarantee that $\rho^{(n_0)}$ is close to $1$. If 
$T\geq (n_0+1)M$, lifted Gabidulin codes can approach $\bar 
C_{\text{CT}}$.  The unknown part is $T<(n_0+1)M$, for which we do not 
know if lifted rank-metric codes achieve $C_{\text{CT}}$.

\subsection{Insufficiency of Unit-block Lifted Rank-Metric Codes} 
\label{sec:oneshot} 
 
In general, we need $n>1$ to have $\rho^{(n)}$ close to $1$. We will 
not study problems like that how large $n$ is sufficient to have 
$1-\rho^{(n)}<\epsilon$ in this paper. But we want to see whether 
$n=1$ is good enough because of its low encoding/decoding 
complexity. As we show in the follows, however, unit-length lifted 
rank-metric codes cannot achieve $\bar C_{\text{CT}}(H,T)$ in general 
and the gap between the maximum achievable rate of unit-block lifted 
rank-metric codes to $\bar C_{\text{CT}}(H,T)$ can be large.  Our 
evaluation reflects the performance of such codes for random linear 
network coding. 
 
Recall that $N^*=\min\{M,N\}$.  For $0<c\leq 
1$ and $N^*>0$ define 
\begin{align*} 
   \rho_{\min}(c,N^*) & = \min_{p_{\rank(H)}:\E[\rank(H)] = c, \rank(H)\leq N^*} \rho^{(1)}. 
\end{align*} 
Considering $T\geq 2M$, there exists a rank distribution of $H$ such that 
\begin{equation*} 
  \max_{\lc\subset \ffield^{(T-M)\times M}}TP_{\text{RM}}(\lc) = \rho_{\min}(c,N^*) \bar C_{\text{CT}}(H,T).  
\end{equation*} 
Linear programming algorithms can be applied to find 
$\rho_{\min}(c,N^*)$.  In Table~\ref{table:w}, we show the values 
$\rho_{\min}(c,6)$ for $c=1,\cdots, 6$. We see $\rho_{\min}(6,6)=1$, 
which is the case that the channel has a constant rank. For $c<6$, 
$\rho_{\min}(c,6)$ is less than $0.65$.  
In Fig.~\ref{fig:v}, we show that the value of 
$\rho_{\min}(3,N^*)$ decreases with $N^*$. $\rho_{\min}(3,200)$ is 
even less than one-fifth, which means that unit-block lifted rank-metric 
codes can achieve less than one-fifth of $\bar{C}_{\text{CT}}(H,T)$.

\begin{table}[!t] 
\renewcommand{\arraystretch}{1.3} 
\caption{The values $\rho_{\min}(c,6)$} 
    \label{table:w} 
\centering 
\begin{tabular}{c||c|c|c|c|c|c} 
  \hline 
  $c$ &  1 & 2 & 3 & 4 & 5 & 6\\ 
  \hline\hline 
  $ \rho_{\min}(c,6) $ & 0.408 & 0.408 & 0.460 & 0.526 & 
  0.649 & 1.0 \\  
  \hline 
\end{tabular} 
\end{table} 
 
\begin{figure}[t] 
  \centering 
  \includegraphics[width=\textwidth]{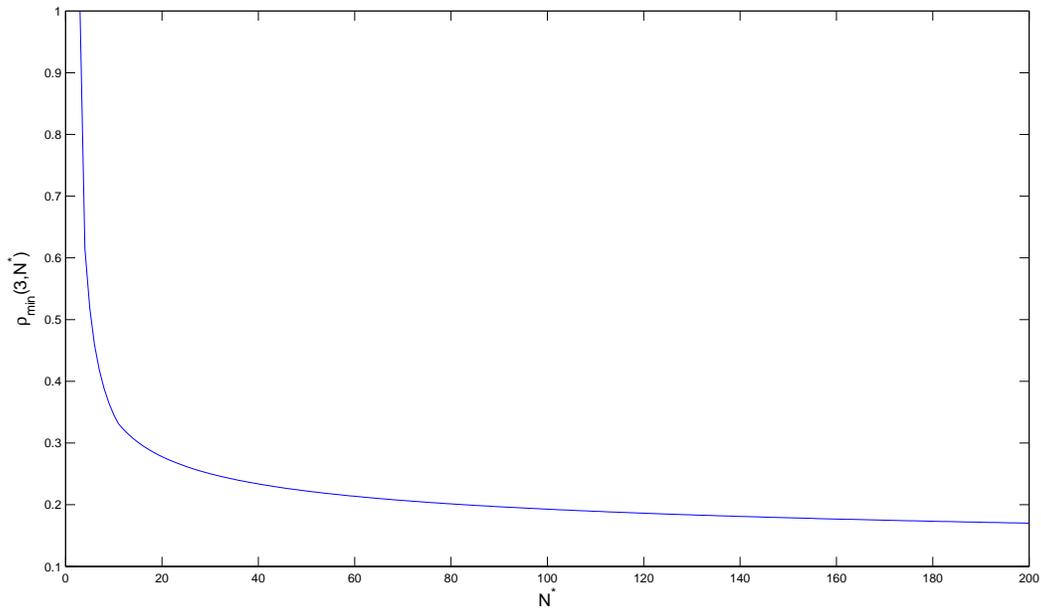} 
  \caption{The value of $\rho_{\min}(3,N^*)$ for 
    $N^*=3,4,\cdots, 200$.}  
  \label{fig:v} 
\end{figure} 
 
\subsection{Complexity of Lifted Rank-Metric Codes} 
 
If we apply Gabidulin codes, a family of MRD codes, the encoding 
requires $\bigO((T-M)n^2sM)$ operations in $\ffield$. For decoding, we 
can apply the algorithm in \cite{silva08j}, the complexity of decoding 
algorithm is given by $\bigO(D(\lc^{(n)})(T-M)n^2s^2)$ operations in 
$\ffield$. (Here we consider that one field operation in $GF(q^m)$ 
require $\bigO(m^2)$ field operations in $\ffield$.)

\section{Linear Matrix Codes  for LOCs} 
\label{sec:linear} 
 
In general, we require $T>> M$ to achieve $\bar C_{\text{CT}}(H,T)$ 
using lifted rank-metric codes. In this section, we propose another 
coding scheme that can achieve $\bar C_{\text{CT}}$ for all $T\geq M$. 
 
\subsection{Linear Matrix Codes} 
 
Consider $\loc(H,T)$ with dimension $M\times 
N$. 
For any positive real number $s\leq N$, let $\mathbf{G}^{(n)}$ be an 
$\floor{ns}\times nM$ matrix, called the generator matrix. The matrix 
code generated by $\mathbf{G}^{(n)}$ is 
\begin{equation*} 
  \mathcal{G}_{T-M}^{(n)} = \{\mathbf{BG}^{(n)}: \mathbf B \in 
  \ffield^{(T-M)\times \floor{ns}} \}. 
\end{equation*} 
The code for $\loc(H,T)$ is the lifted matrix code 
$L(\mathcal{G}_{T-M}^{(n)})$, called \emph{lifted linear matrix code}.  
The rate of $L(\mathcal{G}^{(n)})$ is 
\begin{align*} 
  \mathcal{R}^{(n)} & = 
  \frac{\log_2|\ffield^{(T-M)\times \floor{ns}}|}{nT\log_2 q} \\ 
  & = (1-M/T)\floor{ns}/n. 
\end{align*} 
When $n\rightarrow \infty$, $\mathcal{R}^{(n)} \rightarrow 
(1-M/T)s$.

Suppose that the transmitted codeword is 
$L(\mathbf{B}_0\mathbf{G}^{(n)})$. 
The received matrix is given by \eqref{eq:output}.  
The decoding equation in \eqref{eq:dec} now becomes 
\begin{equation}\label{eq:declinear} 
  \tilde \bY^{(n)} = \mathbf{B}_0\mathbf{G}^{(n)}\bH^{(n)}. 
\end{equation} 
Since the receiver knows $\bH^{(n)}$ and 
$\mathbf{G}^{(n)}$, the information $\mathbf{B}_0$ can be uniquely determined if and only if $\mathbf{G}^{(n)}\bH^{(n)}$ is full rank.  
Thus, the decoding error $P_{e}^{(n)}$ using \eqref{eq:declinear} satisfies 
\begin{equation*} 
  P_{e}^{(n)} \leq \Pr\{\rank(\mathbf{G}^{(n)}H^{(n)}) < \floor{ns}\}. 
\end{equation*} 
 
We can prove the following result in the next subsection.  
 
\begin{theorem}\label{the:linear} 
  Consider linear matrix codes for $\loc(H,T)$ with dimension $M\times 
  N$, and $(s, \epsilon)$ satisfying $0<s < s + \epsilon <\E[\rank(H)]$.   
  More than half of the matrices $\mathbf{G}^{(n)}\in 
  \ffield^{\floor{ns}\times nM}$, when used as the generator matrix, 
  give that 
  \begin{equation*} 
    P_{e}^{(n)} \leq \Pr\{\rank(\mathbf{G}^{(n)}H^{(n)}) < \floor{ns}\} < 2\left(\frac{q^{-\floor{n\epsilon}}}{q-1} + g{(s+\epsilon)}^{n}\right) 
  \end{equation*} 
  where $g{(s+\epsilon)}<1$ is defined in \eqref{eq:fung}. 
\end{theorem} 
 
Thus for any $R<\E[\rank(H)]$, there exists a sequence of lifted linear matrix codes 
with rate $\mathcal{R}^{(n)}\rightarrow R$ and $P_e^{(n)}\rightarrow 0$ 
as $n\rightarrow \infty$. Moreover, $P_e^{(n)}$ decreases exponentially 
with the increasing of $n$.  
So lifted linear matrix codes can achieve the rate $(1-T/M)\E[\rank(H)]$.

\subsection{Performance of Linear Matrix Codes}

\begin{lemma}[Chernoff Bound]\label{lemma:chernoff} 
  Let $\tau_i$, $i=1,\cdots,n$, are independent random variables with 
  the same distribution of $\tau\in\{0,1,\cdots,m\}$. For $\alpha<\E[\tau]$,  
  \begin{equation*} 
    \Pr\left\{\sum_{i}\tau_i < n\alpha\right\} \leq g(\alpha)^n, 
  \end{equation*} 
  where 
  \begin{align} 
    g(\alpha) & = \E[ (A/B)^{(\tau-\alpha)/m}] <1, \label{eq:fung}\\ 
    A & = \sum_{r<\alpha} (\alpha - r) p_\tau(r), \nonumber\\ 
    B & = \sum_{r>\alpha} (r - \alpha) p_\tau(r). \nonumber 
  \end{align} 
\end{lemma} 
\begin{IEEEproof} 
  For any $t>0$, 
  \begin{align} 
    \Pr\left\{\sum_{i}\tau_i < n\alpha\right\}  
    & = \Pr\left\{e^{-t\sum_{i}\tau_i} >e^{-t n\alpha}\right\}\nonumber \\  
    & \locequa{\leq} e^{t n\alpha}\E[e^{-t\sum_{i}\tau_i}] \nonumber\\  
    & \locequa{=} e^{t n\alpha}\prod_i\E[e^{-t\tau_i}] \nonumber \\   
    & = \left(e^{t\alpha}\E[e^{-t\tau}]\right)^n,  \label{eq:ag3}  
  \end{align} 
  where (a) follows from Markov's inequality and (b) follows from independence.  
   
  Now assume $\alpha<\E[\tau]$. 
  Let $f(t)=e^{t\alpha}\E[e^{-t\tau}]$. We know that $f(t)$ is a continuous function for $t\geq 0$ and $f(0)=1$. The first and the second derivatives of $f(t)$ are 
  \begin{align*} 
    f'(t) & = \sum_{r} (\alpha - r) e^{t(\alpha-r)} p_\tau(r),\quad \text{and}\\ 
    f''(t) & = \sum_{r} (\alpha - r)^2 e^{t(\alpha-r)} p_\tau(r), 
  \end{align*} 
  respectively. We see that $f'(0)= \alpha- \E[\tau]<0$ and 
  $f''(t)>0$. Thus, there exists $t_0>0$ such that $f'(t_0)=0$ and 
  $f'(t)<0$ for $0\leq t <t_0$. We give a bound on $t_0$ in the following.  
  Let 
  \begin{align*} 
    A(t) &  = \sum_{r<\alpha} (\alpha - r) e^{t(\alpha-r)} p_\tau(r),\quad \text{and}\\ 
    B(t) & = \sum_{r>\alpha} (r - \alpha) e^{t(\alpha-r)} p_\tau(r). 
  \end{align*} 
  We see that $A(t)$ and $B(t)$ are monotonically increasing and 
  decreasing, respectively. Since $f'(t)=A(t)-B(t)$, we have 
  $A(t_0)=B(t_0)$ and $A(0)<B(0)$. Observe that 
  \begin{align*} 
    A(t) & \leq A(0)e^{t\alpha}, \\ 
    B(t) & \geq B(0)e^{-t(M-\alpha)}. 
  \end{align*} 
  Let $t_1$ such that 
  \begin{equation}\label{eq:892d} 
    A(0)e^{t_1\alpha}=B(0)e^{-t_1(M-\alpha)} 
  \end{equation} 
  We know that $0<t_1\leq t_0$. Thus, $f(t_0)\leq f(t_1)<1$. 
   
  By \eqref{eq:ag3}, 
  \begin{align*} 
    \Pr\left\{\sum_{i}\tau_i < n\alpha\right\} & \leq \min_t f^n(t) \\ 
    & = f^n(t_0) \\ 
    & \leq f^n(t_1). 
  \end{align*} 
  Using \eqref{eq:892d} we have $e^{t_1}= (B(0)/A(0))^{1/M}$. The 
  proof is completed by letting 
  $g(\alpha) = f(t_1)$. 
\end{IEEEproof} 
 
\textbf{Remark:}  
An alternative to the Chernoff bound is Hoeffding's inequality, which gives 
  \begin{equation*} 
    \Pr\left\{\sum_{i}\tau_i < n\alpha\right\} \leq \exp \left\{ -n \frac{2(\alpha-E[\tau])^2}{m^2} \right\}. 
  \end{equation*} 
But in our simulation, the error exponent obtained by the Chernoff bound is better than the one obtained using Heoffding's inequality.

\begin{lemma}\label{lemma:full} 
  Suppose that $G^{(n)}$ is an $\floor{ns}\times nM$ purely random 
  matrix and independent with $H^{(n)}$.  For any $s$ and $\epsilon$ 
  such that $0<s < s + \epsilon <\E[\rank(H)]$,  
  \begin{equation*} 
    \Pr\{\rank(G^{(n)}H^{(n)}) < \floor{ns}\} < \frac{q^{-\floor{n\epsilon}}}{q-1} + g{(s+\epsilon)}^{n}, 
  \end{equation*} 
  where $g(s+\epsilon)<1$ is defined in \eqref{eq:fung}. 
\end{lemma} 
\begin{IEEEproof} 
  Let $F^{(n)}=G^{(n)}H^{(n)}$ and let  
  \begin{equation*} 
   a_n(i) \triangleq 
  \Pr\left\{\rank(F^{(n)}) = \floor{ns}| \rank(H^{(n)}) = i \right\}.  
  \end{equation*} 
  Let $F_i$ be the $i$th row of $F^{(n)}$.  Since $G^{(n)}$ contains 
  uniformly independent components, $F_i$, $i=1,\cdots, \floor{ns}$, 
  are independent and uniformly distributed in the vector space 
  spanned by the row vectors of $H^{(n)}$.  For $i\geq 
  \floor{n(s+\epsilon)}$, 
  \begin{align*} 
    a_n(i) & \stackrel{\text{(a)}}{=} \cmatt{i}{\floor{ns}}\\ 
    & =  \prod_{k=i-\floor{ns}+1}^{i} (1-q^{-k}) \\ 
	& > \prod_{k=\floor{n(s+\epsilon)}-\floor{ns}+1}^{\infty} (1-q^{-k}) \\ 
	& \geq \prod_{k=\floor{n\epsilon}+1}^{\infty} (1-q^{-k}) \\ 
	& \geq 1 - \sum_{k=\floor{n\epsilon}+1}^{\infty} q^{-k} \\ 
        & = 1 - q^{-\floor{n\epsilon}}/(q-1), 
  \end{align*} 
  where (a) follows from Lemma~\ref{lm:rk}. 
  Moreover, using the Chernoff bound in Lemma~\ref{lemma:chernoff},  
  \begin{align*} 
    \Pr\{\rank(H^{(n)})  < \floor{n(s+\epsilon)} \} & \leq \Pr\{\rank(H^{(n)}) < n(s+\epsilon) \}  \\ & \leq \left(g(s+\epsilon)\right)^{n}, 
  \end{align*} 
  where $g(\cdot)$ is defined in \eqref{eq:fung} and 
  $g{(s+\epsilon)}<1$.  Therefore, 
  \begin{align*} 
    \lefteqn{\Pr\{\rank(F^{(n)}) = \floor{ns}\}} \\   
     & \geq \sum_{i\geq \floor{n(s+\epsilon)}} a_n(i)p_{\rank(H^{(n)})}(i), 
   \\ &  > \left(1 - \frac{q^{-\floor{n\epsilon}}}{q-1} \right)  
	\Pr\{\rank(H^{(n)}) \geq \floor{n(s+\epsilon} \} \\ 
        & \geq \left(1 - \frac{q^{-\floor{n\epsilon}}}{q-1} \right) \left(1-g(s+\epsilon)^{-n} \right) \\ 
        & > 1 - \frac{q^{-\floor{n\epsilon}}}{q-1} - g{(s+\epsilon)}^{n}. 
  \end{align*} 
The proof is completed. 
\end{IEEEproof}

\begin{lemma}\label{lemma:tye} 
   Let $0\leq b_i\leq 1$, $i=1,\cdots, n$, be a sequence of real numbers. If $\sum_{i=0}^n b_i /n \leq \epsilon/2$ for some $\epsilon>0$, then there are more than half of the numbers in the sequence with values at most $\epsilon$.  
\end{lemma} 
\begin{IEEEproof} 
  Let $\mathcal{A} = \{b_i:b_i\leq \epsilon\}$. If $|\mathcal{A}|\leq n/2$, then 
  \begin{align*} 
    \sum_{i=0}^n b_i & = \sum_{i\in \mathcal{A}} b_i + \sum_{i\notin \mathcal{A}} b_i \\ & > \epsilon(n-|\mathcal{A}|) \\ 
	& \geq n\epsilon/2. 
  \end{align*} 
  We have a contradiction to $\sum_{i=0}^n b_i /n \leq \epsilon/2$. Thus, $|\mathcal{A}|> n/2$. 
\end{IEEEproof}

\begin{IEEEproof}[Proof of Theorem \ref{the:linear}] 
  There are totally $q^{n\floor{ns}M}$ $\floor{ns}\times nM$ matrices. The average probability of error, when using these matrices uniformly, is upper bounded by 
  \begin{align*} 
    \lefteqn{\sum_{\mathbf{G}^{(n)}\in \ffield^{\floor{ns}\times nM}} 
    \Pr\{\rank(\mathbf{G}^{(n)}H^{(n)}) < \floor{ns}\} 
    q^{-n\floor{ns}M}} \\ 
    & = \sum_{\mathbf{G}^{(n)}\in 
      \ffield^{\floor{ns}\times nM}} \Pr\{\rank(\mathbf{G}^{(n)}H^{(n)}) < \floor{ns}\} p_{G^{(n)}}(\mathbf{G}^{(n)}) \\ 
    & = \Pr\{\rank(G^{(n)}H^{(n)}) < \floor{ns}\} \\ & \leq \frac{q^{-\floor{n\epsilon}}}{q-1} + 
    g{(s+\epsilon)}^{n}, 
  \end{align*} 
  where $G^{(n)}$ is a purely random matrix and the last inequality 
  follows from Lemma~\ref{lemma:full}.  Thus by Lemma~\ref{lemma:tye}, half of these matrices give a probability lower than $2(\frac{q^{-\floor{n\epsilon}}}{q-1} + 
    g{(s+\epsilon)}^{n})$. 
\end{IEEEproof} 
 
\subsection{Complexity of Lifted Linear Matrix Codes} 
 
In practice, we can use a pseudorandom generator to generate matrix 
$\mathbf{G}^{(n)}$, called \emph{pseudorandom generator matrix}, and 
share the pseudorandom generator in both the transmitter and the 
receiver. Discussion of the pseudorandom generator design is out of 
the scope of this paper.  The encoding complexity using a pseudorandom 
generator matrix is $\bigO((T-M)Msn^2)$ and the decoding based on 
Gaussian elimination requires $\bigO(n^3s^3+(T-M)n^2s^2)$ operations in 
$\ffield$. 
 
Compared with the lifted Gabidulin Codes, the complexity of decoding a 
lifted linear Matrix code using Gaussian elimination is higher. To 
reduce the complexity of encoding and decoding is an important future 
work to make lifted linear matrix codes practical.

\subsection{Rateless Coding} 
 
Our coding schemes, both the lifted rank-metric codes1 and the lifted 
linear matrix codes, require only $\E[\rank(H)]$.  Here we show that 
the lifted linear matrix codes can be realized ratelessly without the 
knowledge of $\E[\rank(H)]$ if there exists one-bit feedback from the receiver 
to the transmitter.

Suppose that we have a sequence of $R\times M$ matrices $\mathbf{G}_i$, 
$i=1,2,\cdots$, called the series of the generator matrices of rateless lifted linear matrix codes, which is known by both the transmitter and 
the receiver. Here $R$ is a design parameter.  
Write 
\begin{equation*} 
  \mathbf{G}^{(n)} = \begin{bmatrix}\mathbf{G}_1& \mathbf{G}_2&\cdots&\mathbf{G}_n \end{bmatrix}. 
\end{equation*} 
The transmitter forms its messages into a $(T-M)\times R$ 
message matrix $\mathbf{B}$, and it keeps on transmitting 
$L(\mathbf{BG}_i)$, $i=1,2,\cdots$, until it receives a feedback from the receiver.  
The $i$th output of the channel is given in \eqref{eq:output}.  
After collecting the $n$th output, the receiver checks that  
if $\mathbf{G}^{(n)}\bH^{(n)}$ has rank $R$.  
If $\mathbf{G}^{(n)}\bH^{(n)}$ has rank $R$, the receiver sends a feedback to the transmitter and decodes the message matrix $\mathbf{B}$ by solving the equation  
$\tilde\bY^{n} = \mathbf{B} \mathbf{G}^{(n)}\bH^{(n)}$. 
After received the feedback, the transmitter can transmit 
another message matrix.

Applying Theorem~\ref{the:linear}, we can evaluate the performance of the rateless code.  
The rateless lifted linear matrix codes can achieve the rate $(1-M/T)\E[\rank(H)]$. 
 
\begin{corollary}\label{the:rateless}  
  Consider rateless linear matrix codes for $\loc(H,T)$ with dimension 
  $M\times N$. There exists a series of generator matrices of rateless 
  lifted linear matrix code $\mathbf{G}_i\in \ffield^{R\times M}, 
  i=1,2,\dots$, such that the transmission of one message matrix can 
  be successful decoded with probability at least 
  \begin{equation*} 
    1- 2\left(\frac{q^{-\floor{n\epsilon}}}{q-1} + g{(R/n+\epsilon)}^{n}\right) 
  \end{equation*} 
  after $n > R/\E[\rank(H)]$ transmission, where $0<\epsilon<\E[\rank(H)-R/n$ and $g{(R/n+\epsilon)}<1$ is defined in \eqref{eq:fung}. 
\end{corollary}

\section{Concluding Remarks} 
\label{sec:con} 
 
Linear operator channel is a general channel model that including 
linear network coding as well as the classical $Z$-channel as special 
cases. We studied LOCs with general distributions of 
transformation matrices. 
 
This work showed that the expectation of the rank of the transformation 
matrix $\E[\rank(H)]$ is an important parameter of 
$\loc(H,T)$. Essentially, this is the best rate that noncoherent 
transmission can asymptotically achieve when $T$ goes to infinity. 
We show that both subspace coding and channel training can 
achieve at least $(1-M/T)\E[\rank(H)]$.

This work studied subspace coding from an information theoretic point 
of view. Compared with general subspace coding, constant-dimensional 
subspace coding can achieve almost the same rate. Given a LOC, we 
determined the maximum achievable rate of using constant-dimensional 
subspace coding, as well as the optimal dimension.

We determined the maximum achievable rate of using channel training. 
The advantage of subspace coding over channel training in terms of 
rates is not significant for typical channel parameters. So 
considering channel training for LOCs is sufficient for most 
scenarios.  We proposed two coding approaches for LOCs based on channel 
training and evaluate their performance.

Many problems about LOCs need further investigation. For small $T$ 
(e.g., $T\leq M$), we are still lack of good bounds and coding 
schemes.  It is possible to extend this work to LOCs with additive 
errors and multi-user communication scenarios. Moreover, efficient 
encoding and decoding algorithms for the coding approaches we proposed 
are required for practical applications.

\appendices 
 
\section{Counting} 
\label{sec:count}

Parts of the counting problems here can be found 
in various sources, e.g., \cite{wiki, cooper00} and 
reference therein. Here we give the self-contained proofs. 
 
\begin{lemma}\label{lm:NumFullRankM}  
   When  $0\leq r\leq m$, $|\Fr(\ffield^{m\times r})|=\cmat{m}{r}$ 
\end{lemma}  
\begin{IEEEproof} 
  The lemma is trivial for $r=0$, so we consider $r>0$. 
  We can count the number of full rank matrices in 
  $\ffield^{m\times r}$ by the columns. For the first column, we 
  can choose all vectors in $\ffield^m$ except the zero 
  vector. Thus we have $q^m-1$ choices. Fixed the first 
  column, say $v_1$, we want to choose the second column 
  $v_2$ in $\ffield^m$ but is linear independent with 
  $v_1$. Hence, we have $q^m-q$ choices of $v_2$. Repeat 
  this process, we can obtain that the number of full rank 
  $m\times r$ matrices is 
  $(q^m-1)(q^m-q)\cdots(q^m-q^{r-1})= \cmat{m}{r}$. 
\end{IEEEproof} 
 
Recall 
\begin{equation*}
\cmatt{m}{r}=\left\{\begin{array}{ll} 
    (1-q^{-m})(1-q^{-m+1})\cdots(1-q^{-m+r-1}) & 
    r>0 \\ 1 & r=0 \end{array} \right. 
\end{equation*} 
for $r\leq m$. 
 
\begin{lemma}\label{lm:rk} 
  Let $G$ be an $s\times m$ random matrix 
  with uniformly independent components over $\ffield$. Then 
  for $r\leq m$, 
  \begin{equation*} 
    p_{\rank(GH)|\rank(H)}(s|r) = \cmatt{s}{r}, 
  \end{equation*} 
  where $H$ is any $m\times n$ random matrix. 
\end{lemma} 
\begin{IEEEproof} 
  Fix an $m\times n$ matrix $\bH$ with $\rank(\bH) = r$.  
  Let $F = G\bH$ and let $g_i$ and $f_i$ be the $i$th row of 
  $G$ and $F$, respectively.  
  Since $g_i$ contains uniformly independent components, 
  \begin{equation*} 
    \Pr\{ {g_i} = \mathbf{g}\} = q^{-m}. 
  \end{equation*} 
  For $\mathbf{f}$ with 
  $\mathbf{f}^\tr \in \lspan{\bH^\tr}$, 
  \begin{align*} 
    \Pr\{g_i\bH = \mathbf{f}\} & = q^{-m} |\text{Ker}(\bH)| \\ 
    & = q^{-r}, 
  \end{align*} 
  where $\text{Ker}(\bH) = \{\mathbf{g}: \mathbf{gH} = 
  \bzero\}$ and $|\text{Ker}(\bH)| = q^{m-\rank(\bH)}$.  
  So for $\mathbf{F}$ with $\lspan{\mathbf{F}^\tr} \leq \lspan{\bH^\tr}$, 
  \begin{align} 
    p_{GH|H}(\mathbf{F}|\bH) & = \Pr\{g_i\bH = 
    \mathbf{f}_i, i= 1,\cdots, s\} \nonumber \\ 
    & = \prod_{i=1}^s \Pr\{g_i\bH = \mathbf{f}_i\} \nonumber 
    \\  & = q^{-sr}. \label{eq:puncture} 
  \end{align} 
  Thus, 
  \begin{align*} 
    p_{\rank(GH)|H}(s|\bH) & = q^{-mr}|\{\mathbf{F}: 
    \lspan{\mathbf{F}^\tr}\leq \lspan{\bH^\tr}, 
    \rank(\mathbf{F})=s\}| \\ 
    & = q^{-mr} \cmat{r}{s} \\ 
    & = \cmatt{r}{s}, 
  \end{align*} 
  where $|\{\mathbf{F}:\lspan{\mathbf{F}^\tr}\leq \lspan{\bH^\tr}, 
    \rank(\mathbf{F})=s\}| = \cmat{r}{s}$ follows from  
   Lemma~\ref{lm:NumFullRankM}. 
  Last, since $\rank(H)\rightarrow H\rightarrow \rank(GH)$ 
  forms a Markov chain, 
  \begin{align*} 
    p_{\rank(GH)|\rank(H)}(s|r) & = \sum_{\bH:\rank(\bH) = 
      r} p_{\rank(GH)|H}(s|\bH)p_{H|\rank(H)}(\bH|r) \\ 
    & = \cmatt{r}{s} \sum_{\bH:\rank(\bH) = 
      r} p_{H|\rank(H)}(\bH|r) \\  
    & = \cmatt{r}{s}. 
  \end{align*} 
  The proof is complete. 
\end{IEEEproof}

\begin{lemma} \label{lemma:g9ds} 
The number of $r$-dimensional subspace in $\ffield^m$ is 
given by the \emph{Gaussian binomials}.  
\end{lemma}  
\begin{IEEEproof} 
  Define an equivalent relation on 
  $\mathcal{M}(\ffield^{m\times r})$ by $\bX \sim \bX'$ if $\lspan{\bX} = 
  \lspan{\bX'}$. The equivalent class $[\bX]$ is the set of 
  all matrices that equivalent to $\bX$.  
  We have $[\bX] = \{\bX\Phi:\Phi\in 
  \mathcal{M}(\ffield^{r\times r})\}$.  
  Thus $|[\bX]| = |\mathcal{M}(\ffield^{r\times r})| = 
  \cmat{r}{r}$.  
  Since  $\Gr(r,\ffield^T) = \mathcal{M}(\ffield^{m\times 
    r})/\sim$, the quotient set of $\mathcal{M}(\ffield^{m\times 
    r})$ by $\sim$, we have $|\Gr(r,\ffield^T)| = 
  |\mathcal{M}(\ffield^{m\times r})|/|[\bX]| = 
  {\cmat{m}{r}}/{\cmat{r}{r}}$.  
\end{IEEEproof}

\begin{lemma} \label{lm:NumMGivenRank}  
For $m\geq r'$ and $r \geq r'$, define a set $S = \{\bX \in  
\ffield^{m\times r}:\rank(\bX) = r'\}$. Then  
\begin{equation} \label{eq:38dfajj} 
   |S| = \frac{\cmat{m}{r'} \cmat{r}{r'}}{\cmat{r'}{r'}} =\cmat{m,r}{r'}.  
\end{equation}  
Furthermore, 
\begin{equation} 
  \label{eq:g89a} 
  \sum_{r'}\cmat{m,r}{r'} = q^{mr}. 
\end{equation} 
\end{lemma}  
\begin{IEEEproof}  
The column vectors of $\bX\in S$ span an $r'$-dimensional subspace in  
a $m$-dimensional vector space. Let $\{ {V}_1,  {V}_2, \ldots  
 {V}_n\}$ be the set of $r'$-dimensional subspace in a  
$m$-dimensional vector space, where $n = \gco{m}{r'}$. Let  
$S_{ {V}_i}=\{\bX\in \ffield^{m\times r}: \lspan{\bX}= {V}_i\}$ and  
the set $\{S_{ {V}_i}\}$ is a partition of $S$.  By  
 $|\{S_{ {V}_i}\}| = \cmat{r}{r'}$.  
Therefore,  
\begin{equation}  
  \label{eq:49}  
  |S| = n|S_{ {V}_i}|=\gco{m}{r'} \cmat{r}{r'}= \cmat{m,r}{r'}.  
\end{equation}  
The equality in (\ref{eq:g89a}) follows because both sides 
are the number of $m\times r$ matrices. 
\end{IEEEproof}

\begin{lemma}\label{lemma:8balfa} 
  Let $V \leq \ffield^m$ be a $s$-dimensional 
  subspace. Then, the number of subspace $U$ with  $V \leq 
  U$ and $\dim(U)=r$ is 
 \begin{align}\label{eq:9ga9d9d} 
  \gco{m-s}{r-s} = \gco{m}{r}\frac{\cmat{r}{s}}{\cmat{m}{s}}. 
 \end{align} 
\end{lemma} 
\begin{IEEEproof} 
  Let $U$ be a subspace with $V\leq U$ and 
  $\dim(U)=r$. Then we can write $U=V+U'$ where $U'$ is a 
  $\dim(U')=r-s$ and $V\cap U'=\{0\}$. Given $U$, such $U'$ 
  is unique. The number of $U'$ is the number of 
  $(r-s)$-dimensional subspace in an $(m-s)$-dimensional space, 
  i.e., $\gco{m-s}{r-s}$. The equality in (\ref{eq:9ga9d9d}) 
  is the direct result of the definitions. 
\end{IEEEproof} 
 
\section{Useful Results} 
 
\begin{lemma} \label{lm:eight} 
For $r\leq m$, 
  $- \log_2 \cmatt{m}{r} < 1.8$. 
\end{lemma} 
\begin{IEEEproof} 
Define 
\begin{equation}\label{eq:xi} 
  \Xi_q(s) = \prod_{i= s}^{\infty} (1-q^{-i}). 
\end{equation} 
So $\cmatt{m}{r} > \Xi_q(m-r+1)$.  
We know $\Xi_q(s+1)>\Xi_q(s)>\Xi_{q-1}(s)\geq \Xi_2(1)$, 
where $\Xi_2(1)$ is a mathematics constant with approximate 
value $0.28879$ \cite{cooper00}. 
Thus $- \log_2 \cmatt{m}{r} \leq -\log_2 \Xi_2(1) < - \log_2 0.2887 <1.8$. 
\end{IEEEproof}

\begin{lemma}\label{lemma:car} 
  $\lim_{T\rightarrow \infty} \frac{\log_2 \cmat{T}{r}}{T\log_2 q} = r$. 
\end{lemma} 
\begin{IEEEproof} 
  \begin{align*} 
    \lim_{T\rightarrow \infty} \frac{\log_2 \cmat{T}{r}}{T\log_2 q} & =  
    \lim_{T\rightarrow \infty} \frac{\log_2 \cmatt{T}{r} q^{Tr}}{T\log_2 
      q} \\ & = \lim_{T\rightarrow \infty} \frac{\log_2 
      \cmatt{T}{r}}{T\log_2 q}  + \lim_{T\rightarrow \infty} \frac{\log_2 
      q^{Tr}}{T\log_2 q} \\ & = 0 + r. 
  \end{align*} 
\end{IEEEproof} 
 
\begin{lemma}\label{lemma:pjn} 
  $|\Pj(\ffield^m)|< q^{m^2/2 + \log_q m + c}$, where $c< 1.8$ is 
  a constant.  
\end{lemma} 
\begin{IEEEproof} 
Refer to the proof of Lemma~\ref{lm:eight}. We have 
  \begin{align*} 
    |\Pj(\ffield^m)| & = \sum_{r\leq m} \gcos{m}{r}  \\ 
    & = \sum_{r\leq m} q^{(m-r)r} 
    \frac{\cmatt{m}{r}}{\cmatt{r}{r}} \\ 
    & < \sum_{r\leq m} q^{(m-r)r} 
    \frac{1}{\Xi_q(1)} \\ 
    & < \frac{m}{\Xi_q(1)} q^{m^2/2}\\ 
    & = q^{m^2/2 + \log_q (m/\Xi_q(1))} \\ 
    & < q^{m^2/2 + \log_qm + \log_2(1/\Xi_2(1))}. 
  \end{align*} 
  Let $c=\log_2(1/\Xi_2(1))$. By  $\Xi_2(1)\approx 0.28879$, we 
  obtain $c<1.8$. 
\end{IEEEproof}

\begin{lemma}\label{lm:oop} 
  For $V \leq U \leq \ffield^T$ and $V' \leq U' \leq \ffield^T$ with 
  $\dim(U)=\dim(U')$ and $\dim(V)=\dim(V')$, we can find $\Phi\in 
  \Fr(\ffield^{T\times T})$ such that $\Phi U=U'$ and $\Phi V = V'$. 
\end{lemma} 
\begin{IEEEproof} 
  Find a basis $\{\mathbf{b}_{i}:i=1,\cdots, T\}$ of 
  $\ffield^T$ such that $\{\mathbf{b}_{i}:i=1,\cdots, r\}$ is a basis of 
  $U$ and $\{\mathbf{b}_{i}:i=1,\cdots, s\}$ is a basis of $V$.  
  We can do this by first finding a basis of $V$, extending the 
  basis to a basis of $U$ and further extending to a basis of $\ffield^T$. 
  Similarly, find a basis $\{\mathbf{b}'_{i}:i=1,\cdots, T\}$ of 
  $\ffield^T$ such that $\{\mathbf{b}'_{i}:i=1,\cdots, r\}$ is a basis of 
  $U$ and $\{\mathbf{b}'_{i}:i=1,\cdots, s\}$ is a basis of $V$.  
  Consider the linear system of equations 
  \begin{equation*} 
    \Phi \mathbf{b}_i = \mathbf{b}_{i}',\quad i=1,\cdots,T. 
  \end{equation*} 
  We know there exists unique $\Phi\in \Fr(\ffield^{T\times T})$ 
  satisfying this linear system and $\Phi V = V'$ and $\Phi U = U'$. 
\end{IEEEproof}

\begin{lemma}\label{lemma:jingle} 
  For $\bX, \bX' \in \ffield^{T\times M}$, 
  $\lspan{\bX^\tr}=\lspan{\bX'^\tr}$ if and only if there exists 
  $\Phi\in\Fr(\ffield^{T\times T})$ such that $\bX' = \Phi \bX$. 
\end{lemma} 
\begin{IEEEproof}
  Let $r = \rank(\bX)$. First, show a) $\Rightarrow$ c).  Fix 
    one full-rank decomposition $\bX = \mathbf B \mathbf 
    D$. Since $\lspan{\mathbf D^\tr} = \lspan{\bX^\tr} = 
    \lspan{\bX'^\tr}$, we can find a decomposition $\bX' = 
    \mathbf B' \mathbf D$ using the same procedure we 
    described by first fixing $\mathbf D$. Second, show c) 
    $\Rightarrow$ b). With the decomposition in c), we can 
    find $\Phi\in \Fr(\ffield^{T\times T})$ such that $\Phi \mathbf B 
    = \mathbf B'$. Extend $\mathbf B$ and $\mathbf B'$ to 
    $T\times T$ matrices $[\mathbf B\ \mathbf B_0]$ and 
    $[\mathbf B'\ \mathbf B_0']$. Then, $\Phi =[\mathbf B'\ 
    \mathbf B_0'] [\mathbf B\ \mathbf B_0]^{-1}$ is one such 
    matrix we want since $\Phi [\mathbf B\ \mathbf B_0] = 
    [\mathbf B'\ \mathbf B_0']$. Last, we have b) 
    $\Rightarrow$ a).  
\end{IEEEproof}

\begin{lemma}\label{lm:ia} 
  For $U\leq \ffield^t$ with $\dim(U)=r\leq m$, 
  let  
  \begin{equation*} 
  A(m,U) = \{\bX\in \ffield^{t\times m}:\lspan{\bX} = U\}. 
\end{equation*} 
  Then,  
  \begin{equation*} 
    |A(m,U)| = \cmat{m}{r}, 
  \end{equation*} 
  and for $\Phi\in \Fr(\ffield^{t\times t})$ 
  \begin{equation*} 
    A(m,\Phi U) = \Phi A(m,U). 
  \end{equation*} 
\end{lemma} 
\begin{IEEEproof} 
Find a $t\times r$ matrix $\mathbf B$ with $\lspan{\mathbf 
  B} = U$. Then, we have 
\begin{equation*}
  A(m,U) = \{\mathbf{BD}:\mathbf D\in \Fr(\ffield^{r\times 
    m})\} =  \mathbf{B} \Fr(\ffield^{r\times m}). 
\end{equation*} 
Thus, $|A(m,U)| = |\Fr(\ffield^{r\times m})| = \cmat{m}{r}$.   
For $\Phi\in \Fr(\ffield^{t\times t})$, $\lspan{\Phi \mathbf B} = \Phi U$. 
So $A(m,\Phi U) = \Phi \mathbf B  \Fr(\ffield^{r\times M}) = \Phi A(m,U)$. 
\end{IEEEproof} 
 
\section*{Acknowledgement} 
 
Shenghao Yang thanks Kenneth Shum for helpful discussion.


\end{document}